\documentclass[%
reprint,
superscriptaddress,
nofootinbib,
amsmath,amssymb,
aps,
pra,
]{revtex4-2}

\usepackage{graphicx}
\usepackage{dcolumn}
\usepackage{mathtools}
\usepackage{amsmath,amsfonts,amssymb,amsthm}
\usepackage{thmtools, thm-restate}
\usepackage{bbm}
\usepackage{xcolor}
\usepackage{tikz}
\usepackage{pgfplots}
\usepgfplotslibrary{fillbetween}
\pgfplotsset{compat=1.18}
\usepackage{float}

\newtheorem{theorem}{Theorem}
\newtheorem{corollary}{Corollary}[theorem]

\usepackage{hyperref}
\usepackage[capitalise]{cleveref}
\usepackage[mathlines]{lineno}

\makeatletter
\AddToHook{cmd/appendix/before}{\def\cref@section@alias{appendix}}

\AddToHook{cmd/appendix/before}{%
    \crefalias{section}{appendix}%
    \crefalias{subsection}{appendix}
}
\makeatother

\usepackage{soul}

\begin{document}
	
	\preprint{APS/123-QED}
	
	\title{Bounds on Nonlocality and Random Access Codes from\\ Extended Information Causality Principle }
	
	\author{Prabhav Jain}
    \email{prabhav.jain@tu-darmstadt.de}
	\affiliation{Department of Computer Science, Technical University of Darmstadt, 64289 Darmstadt, Germany}
	\author{Nikolai Miklin}
    \affiliation{Department of Physics, Technical University of Darmstadt, 64289 Darmstadt, Germany}
	\author{Mariami Gachechiladze}
	\affiliation{Department of Computer Science, Technical University of Darmstadt, 64289 Darmstadt, Germany}

	\date{\today}
	
	\begin{abstract}
		Information Causality was introduced as a physical principle for constraining the set of nonlocal correlations. In recent work, we proposed an extension of Information Causality that allows correlations among Alice’s inputs. This extended principle yields tighter constraints than the original formulation and recovers part of the quantum boundary in certain Bell scenarios. In this work, we further investigate the implications of extended Information Causality and apply it to scenarios beyond binary inputs and outputs. We derive a family of quantum Bell inequalities that strengthen previously known constraints on quantum correlations. Using these inequalities, we obtain an improved analytical bound for the Collins-Gisin family of Bell inequalities. We also apply Information Causality to entanglement-assisted random access codes and derive new theory-independent analytical bounds on the winning probability. For this latter task, we prove that, despite being stronger in general, the extended principle does not improve the bounds obtained from the original Information Causality principle. This suggests that the existing Information Causality bounds are optimal for this class of random access codes.
	\end{abstract}
	\maketitle
	
	
	\section{Introduction}
	
    Since the advent of quantum mechanics, several distinctive features of the theory have been identified, including superposition, uncertainty relations between observables, and entanglement, all of which mark a departure from classical physics. Bell non-locality is another such feature: it allows correlations that are stronger than those achievable by any local classical theory, including theories with hidden variables \cite{bell1964epr,brunner2014bell}. In recent years, Bell non-locality has attracted renewed interest, and increasingly stringent experiments have verified its existence while closing several potential loopholes \cite{bellexperiment1,bellexperiment2,Hensen2015}.

	Although Bell non-locality is well understood within the formalism of quantum mechanics, it is desirable to identify physical or operational principles that explain this behavior without relying on the specific mathematical structure of the theory. In \cite{popescu1994quantum}, the principle of non-signaling (NS) was introduced as a requirement imposed by causality, namely the impossibility of faster-than-light signaling. The authors showed, however, that non-signaling alone is not sufficient to single out quantum correlations: there exist broad classes of post-quantum distributions that satisfy non-signaling. This raises the question of which additional properties, beyond causality, a distribution must satisfy in order to be quantum-realizable.

    To address this question, several physical principles have been proposed, including Macroscopic Locality (ML) \cite{navascues2007bounding}, Local Orthogonality \cite{fritz2013local}, Non-Trivial Communication Complexity (NTCC) \cite{brassard2006limit}, Almost Quantum Correlations (AQC) \cite{navascus2015aqc}, and Information Causality (IC) \cite{pawlowski2009information}. The broad aim of these approaches is to recover, or at least approximate, the set of quantum correlations from empirical or operational principles alone. When several such principles are considered, a natural question is how they compare. Almost Quantum Correlations, which correspond to the $Q_{1+AB}$ level of the Navascués-Pironio-Acín hierarchy and strictly contain the quantum set, are known to satisfy all of the above principles except, potentially, Information Causality \cite{navascus2015aqc}. Thus, the relation between AQC and IC remains unresolved.

    In the original IC paper \cite{pawlowski2009information}, Information Causality was shown to reproduce Tsirelson’s bound for the Clauser-Horne-Shimony-Holt (CHSH) inequality \cite{clauser1969proposed}, as well as Uffink’s inequality \cite{uffink2002quadratic}. Later works applied IC to more general Bell scenarios, in which both parties may have an arbitrary number of measurement settings, and derived a family of quantum Bell inequalities \cite{gachechiladze2022quantum,miklin2021information}. In \cite{our_prl,miklin2026communication}, an extension of the IC principle was proposed for arbitrary distributed-computation scenarios. In the binary case, this extension was shown to be stronger than Uffink’s inequality and to recover part of the quantum boundary in a particular slice of the correlation space, where it coincides with the Tsirelson-Landau-Masanes (TLM) inequality \cite{cirel1980quantum,landau1988empirical,masanes2003necessary}. In \cite{miklin2026communication}, it was also shown that IC is stronger than NTCC, in the sense that any non-signaling distribution satisfying IC also satisfies NTCC. More recently, \cite{pollyceno2026communication} introduced a protocol-independent and task-independent approach to applying IC. In their framework, they re-derived the extended IC inequality presented in Refs.~\cite{our_prl,miklin2026communication}, which is also the inequality used throughout the present work. Their results show that this extended IC inequality can provide stronger bounds than the original IC statement in scenarios beyond random access codes, including more general Bell scenarios and communication games~\cite{ambainis2009quantumrandomaccesscodes,pawowskiearac,acin2007device}. 
    
	In this work, we further investigate the implications of the extended IC principle and derive a family of analytical quantum Bell inequalities for scenarios beyond binary measurement settings and outcomes. We show that these inequalities are strictly tighter than the previously derived family of inequalities in \cite{gachechiladze2022quantum,our_prl}, which were obtained using the original IC principle. Using the new inequalities, we obtain an improved analytical bound on the Collins-Gisin family of Bell inequalities \cite{collins2004relevant}. Thus, we demonstrate that the extended IC principle is stronger than the original formulation in more general Bell scenarios.
	 
	We also apply IC to derive upper bounds on the winning probability of entanglement-assisted random access codes (EARACs) \cite{ambainis2009quantumrandomaccesscodes,pawowskiearac}. We first study the standard EARAC scenario, recovering several known bounds from the literature, and then derive new bounds for higher-dimensional generalizations. We relate these winning probabilities to critical noise-mixing thresholds and, in the final section, consider the effect of correlations among Alice’s inputs. We prove the somewhat surprising result that, although the extended IC principle is tighter in general, it does not yield better bounds than the original IC statement for these tasks. Finally, as a corollary of the above result, we argue for the optimality of the known IC-implied bounds on EARAC winning probabilities.
    
	\section{Preliminaries}\label{sec:setup}
    In this section, we describe the Bell scenarios considered in this work and fix the notation used throughout the paper. For an integer $d\geq 2$, we denote by $[d]:=\{0,\dots,d-1\}$ the alphabet of size $d$. We refer to a variable taking values in $[d]$ as a \textit{dit}; in particular, a bit corresponds to the case $d=2$. Unless stated otherwise, the symbols $\oplus$ and $\cdot$ denote addition and multiplication modulo the relevant alphabet size.

	We consider two parties, Alice and Bob. We use the notation $m_\mathcal{A}m_\mathcal{B}n_\mathcal{A}n_\mathcal{B}$ for the Bell scenario in which Alice and Bob have $m_\mathcal{A}$ and $m_\mathcal{B}$ measurement settings, and $n_\mathcal{A}$ and $n_\mathcal{B}$ measurement outcomes, respectively. In the IC communication scenario, illustrated in \cref{fig:scenario}, Alice receives a string of $n$ dits $(a_0,\dots,a_{n-1})$, each taking values in $[d_\mathcal{A}]$, distributed according to a specified joint probability distribution. Bob receives an index $b\in[n]$, distributed uniformly at random. Alice and Bob share a non-signaling box described by conditional probabilities $\mathrm{P}(A B|\alpha\beta)$, where $\alpha\in[d_\alpha]$ and $\beta\in[d_\beta]$ are the inputs to the box, while $A,B\in[d_O]$ are its outputs. They are also allowed to use a noisy classical communication channel that maps an input message $x\in[d_C]$ to an output message $x'\in[d_C]$, where $d_C$ is the channel alphabet size. Bob's goal is to output a guess $g$ for Alice's $b$th dit $a_b$.

    We begin by recalling the original formulation of Information Causality introduced in \cite{pawlowski2009information}. For the scenario described above, the IC statement is
	\begin{align}
	\sum_{i=0}^{n-1}I(a_i;g \mid b=i)\leq \mathcal{C}. 
	\end{align}
    Here, $I(\cdot;\cdot)$ denotes Shannon mutual information and $\mathcal{C}$ is the capacity of the classical communication channel. Since quantum correlations satisfy IC, the set of non-signaling correlations satisfying IC provides an outer approximation to the quantum set \cite{pawlowski2009information}. The original formulation of IC assumes that Alice's inputs are independent. In recent work \cite{our_prl,miklin2026communication}, this assumption was relaxed, and an extension of IC was proven that applies to arbitrary distributed functions and allows correlations among Alice's inputs. 

    Let $\mathcal{A}$ and $\mathcal{B}$ be finite input sets, and consider a binary-valued function $f:\mathcal{A}\times\mathcal{B}\to\{0,1\}$ computed in a distributed manner. The extended IC statement is
    \begin{equation}\label{eq:eic}
    \sum_{i=0}^{|\mathcal{B}|-1}
    I\!\left(f(a,b^i);g\mid  b=b^i, \{f(a,b^j)\}_{j=0}^{i-1}\right)
    \leq \mathcal{C},
    \end{equation}
    where $b^0,b^1,\dots,b^{|\mathcal{B}|-1}$ is an ordering of the elements of $\mathcal{B}$.
    
    For example, consider the binary index function, $\mathrm{INDEX}_2(a,b)\coloneqq a_b$, with $\mathcal{A}=\{0,1\}^2$, $\mathcal{B}=[2]$. Substituting this function into \cref{eq:eic}, with the ordering $b^0=0$ and $b^1=1$, yields the following generalization of the original IC statement for the binary scenario:
    \begin{align}\label{eq:ic_statement_correlated}
    I(a_0;g \mid b=0)+I(a_1;g \mid b=1,a_0)\leq \mathcal{C}.
    \end{align}
    This generalization leads to a quantum Bell inequality that is stronger than the one obtained from original IC inequality in the binary $(2222)$ scenario. For convenience, we restate this inequality from \cite{our_prl}:
	\begin{equation}\label{eq:better_Uffink}
		\left((1+\epsilon)e_{00}+(1-\epsilon)e_{10}\right)^2+(1-\epsilon^2)(e_{01}-e_{11})^2\leq 4,
	\end{equation}
	where $e_{ji}\coloneqq 2\mathrm{P}(A=B \mid \alpha=j,\beta=i)-1$ are the equality biases of the shared non-signaling box, and $\epsilon$ quantifies the correlation between Alice's input bits. In particular, if $\epsilon=2\mathrm{P}(a_0=a_1)-1$, then $\epsilon=1$ corresponds to perfect correlation. In the subspace of distributions with uniform marginals for both parties, \cref{eq:better_Uffink} is tighter than Uffink's inequality \cite{uffink2002quadratic}. In the limit $\epsilon\to1$, it coincides with the TLM inequality \cite{cirel1980quantum,landau1988empirical,masanes2003necessary}, thereby recovering part of the quantum boundary \cite{our_prl}.
	
    In what follows, we generalize the above derivation to two commonly studied families of Bell scenarios. First, we consider the $nn22$ scenario, in which Alice and Bob each have $n$ measurement settings and binary outcomes. Second, we consider the $d2dd$ scenario, in which Alice has $d$ measurement settings, Bob has two measurement settings, and both parties have $d$ outcomes.
	\begin{figure}
		\centering
		\includegraphics[width=0.8\linewidth]{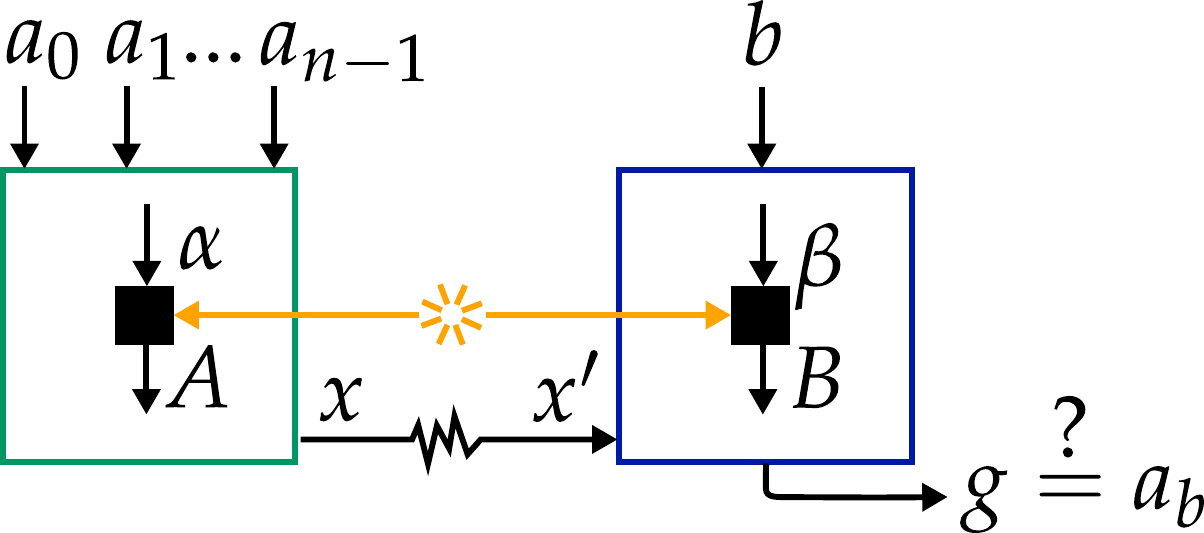}
		\caption{The communication scenario. Alice receives a string of dits $(a_0,\dots,a_{n-1})$, while Bob receives an index $b\in[n]$. Alice and Bob share a non-signaling box and may communicate through a noisy classical channel. Bob's goal is to output a guess for Alice's dit $a_b$.}
		\label{fig:scenario}
	\end{figure}
	
\section{Tighter Quantum Bell Inequalities in the $nn22$ Scenario}\label{sec:nn22}
In the $nn22$ scenario, a family of quantum Bell inequalities generalizing Uffink's inequality was derived from the original IC principle in \cite{our_prl}. In this section, we extend that result and derive new, tighter quantum Bell inequalities based on the extended IC statement. We choose the index function $\mathrm{INDEX}_n(a,b)\coloneqq a_b$, with $\mathcal{A}=\{0,1\}^n$ and $\mathcal{B}=[n]$. Substituting this function into \cref{eq:eic} with the natural ordering of $[n]$, the extended IC statement becomes
\begin{align}\label{eq:ic_corr_nn22}
    \sum_{i=0}^{n-1} I(a_i;g \mid   b=i, \{a_j\}_{j<i})\leq \mathcal{C}. 
\end{align}
    
To derive the quantum Bell inequalities, we follow the procedure described in \cite{our_prl}. We use van Dam's protocol \cite{vandam} and compute the joint distribution of Alice's bits and Bob's guess. From this distribution, we obtain the guessing probability for each bit and then evaluate the mutual information terms in \cref{eq:ic_corr_nn22}. To extract quadratic inequalities from these mutual-information terms, we use the technique outlined in \cite{our_prl}. In the limit in which the capacity of the classical communication channel tends to zero, both sides of \cref{eq:ic_corr_nn22} vanish. Hence, this limit can be evaluated by applying L'Hospital's rule twice, yielding quadratic constraints from the entropic inequality. We state the resulting constraints below.

\begin{restatable}{theorem}{thmnn}
\label{thm:nn22_ineq}
Any non-signaling theory satisfying the extended IC principle complies with the following inequalities
\begin{equation}\label{eq:nn22_ineq}
     \sum_{i=0}^{n-1}2^{i}\sum_{k_{j<i}}\left(\sum_{k_i}(-1)^{k_i}\sum_{k_{j> i}}(-1)^{h(\vec{k})}e_{f(\vec{k}),i}\right)^2\leq 4^n,
\end{equation}
for any functions $f:\{0,1\}^n\to[n]$ and $h:\{0,1\}^n\to\{0,1\}$. Here
 $e_{j,i}\coloneqq 2\mathrm{P}(A=B \mid \alpha=j,\beta=i)-1$, for all $i,j\in [n]$, and $k_i$ denotes the $i$-th component of the $n$-dimensional binary vector $\vec{k}$.
\end{restatable}

Here and below, $\sum_{k_{j<i}}$ denotes summation over all assignments of the variables $\{k_j:j<i\}$, and $\sum_{k_{j>i}}$ is defined analogously.
See \cref{app:nn22_ineqs} for a proof of the above theorem.

We now use these inequalities to upper-bound the value of the $I_{nn22}$ expressions introduced in \cite{collins2004relevant}. In \cite{our_prl}, it was shown that the maximal quantum value of $I_{nn22}$ is upper-bounded by $\frac{n-3}{2}+\sqrt{\frac{1}{3}+\frac{8}{3\cdot 4^n}}$, which is lower than the non-signaling value of $\frac{n-1}{2}$ by approximately $1-\sqrt{\frac{1}{3}}$ for large $n$. That derivation used a quantum Bell inequality obtained from standard IC for a specific protocol. Applying \cref{thm:nn22_ineq} to the same protocol yields a tighter inequality for the scenario and leads to the improved upper bound stated in the following theorem.
\begin{restatable}{theorem}{thmInn}\label{thm:Inn22}
    For any $n\geq 3$, the value of the $I_{nn22}$ expression achieved by any non-signaling box with uniform marginals and satisfying the extended Information Causality principle is upper-bounded by
    \begin{equation}
        \frac{n-4}{2}+\sqrt{\frac{1}{2}+\frac{1}{2^{n-2}}}.
    \end{equation}
\end{restatable}
See \cref{app:inn22_details} for the proof. This bound is always below the non-signaling maximum by at least $\frac{3-\sqrt 2}{2}$ and thus demonstrates an asymptotic separation between the set of quantum and non-signaling correlations. One remarkable fact is that the communication protocol assumes Alice's inputs to be independent. Despite that, we get an improvement over the original bound derived in \cite{our_prl}, which further exhibits the fact that the extended IC statement is inherently tighter and does not necessarily need to rely on correlations among inputs to provide stronger bounds. 

One may ask whether better bounds can be obtained by applying the entropic inequality in \cref{eq:ic_corr_nn22} directly, rather than passing to the zero-capacity regime. For isotropic boxes, where the NS-box biases are independent of both parties' settings, we find that in the $nn22$ scenario the tightest bounds are attained in the limit of vanishing capacity; see \cref{app:iso_nn22} for details. Thus, in this subspace, the quadratic constraints in \cref{thm:nn22_ineq} are as tight as the original entropic inequality.

\section{Tighter Quantum Bell Inequalities in the $d2dd$ Scenario}\label{sec:d2dd}
In the $d2dd$ scenario, a family of quantum Bell inequalities generalizing Uffink's inequality was derived using the original IC principle in \cite{gachechiladze2022quantum, our_prl}. Within a specific correlation subspace, these inequalities coincide exactly with the constraints imposed by Macroscopic Locality~\cite{gachechiladze2022quantum, navascues2010glance}. In this section, we generalize these results and derive new, tighter quantum Bell inequalities using the extended IC principle.

In this scenario, the extended IC statement has the same form as \cref{eq:ic_statement_correlated}, but the parties' inputs and outputs are no longer binary. We specify the relevant probability distributions through their associated biases as follows:
   \begin{align}\label{eq:input_dist}
    \begin{split}
    \mathrm{P}(a_0=k,a_1=l)&=\frac{1+q_{kl}}{d^2},\\
    \mathrm{P}(x'=x\oplus m)&=\frac{1+(d-1)e_{c}^m}{d},\\
    \mathrm{P}(A\oplus B=k\mid \alpha=i,\beta=j)&=\frac{1+(d-1)e^k_{ij}}{d}.
    \end{split}
    \end{align}
The biases are constrained by normalization, namely $\sum_{k,l}q_{kl}=0$, $\sum_m e_c^m=0$, and $\sum_k e^k_{ij}=0$ for every $i,j$. Here, the first expression specifies the joint distribution of Alice's input dits, with biases $\{q_{kl}\}_{k,l=0}^{d-1}$. The second expression describes a symmetric classical channel, namely the probability that the received message $x'$ differs from the sent message $x$ by a shift $m\in[d]$. The channel is therefore specified by the biases $\{e_c^m\}_{m=0}^{d-1}$. Finally, the third expression specifies the distribution of the modular output difference $A\oplus B$ of the non-signaling box, with associated biases $\{e^k_{ij}: i,k\in[d],\,j\in\{0,1\}\}$.

The derivation follows the same steps as the proof of \cref{thm:nn22_ineq}. We state the resulting inequalities below.

    \begin{theorem}\label{thm:corr_ineq}
		In the $d2dd$ scenario, any non-signaling theory satisfying the extended IC principle complies with the following inequalities:
		\begin{align}\label{eq:d2dd_corr_ineq}
        \begin{split}
			&-\sum_{j=0}^{d-1}(h^0_{j})^2+\frac{1}{d}\sum_{j,k=0}^{d-1}\left((h^0_{jk})^2-(h^1_{jk})^2\right) \\
			&+\frac{1}{d^2}\sum_{j,k,l=0}^{d-1}(1+q_{kl})(f^1_{jkl})^2\leq (d-1)^2 \sum_{m=0}^{d-1}(e_c^m)^2.
             \end{split}
		\end{align}
        where the tensor $f$ corresponds to a convolution of the NS-box biases $\{e^k_{ij}\}_{i,k\in [d]}^{j\in \{0,1\}}$ and the channel biases $\{e^m_c\}_{m=0}^{d-1}$ and is given by
        \begin{equation}
			f_{jkl}^i=\frac{(d-1)^2}{d}\sum_{m=0}^{d-1}e^m_ce^{j\oplus\overline{k}\oplus\overline{m}}_{\overline{k}\oplus l,i},
		\end{equation}
         while the tensors $h^i_{jkl}=q_{kl}+(1+q_{kl})f^i_{jkl}$,  $
			  h^i_{jk}=\frac{1}{d}\sum_{l=0}^{d-1}h^i_{jkl}$ and $ h^i_{j}=\frac{1}{d^2}\sum_{k,l=0}^{d-1}h^i_{jkl}$ are contractions of the $f$ tensor over Alice's dits.
	\end{theorem}

The proof is given in \cref{app:d2dd_details}. We remark that, unlike the inequalities derived from the original IC statement in \cite{gachechiladze2022quantum,our_prl}, which are formulated only in terms of the NS-box biases, the inequalities in \cref{eq:d2dd_corr_ineq} also depend on the channel parameters. For an arbitrary non-uniform distribution of Alice's input dits, different channels can therefore yield inequivalent constraints, and there is no single channel choice that is optimal in general.

Next, we compare the inequalities in \cref{eq:d2dd_corr_ineq} with the previous family of quantum Bell inequalities derived using the original IC principle in \cite{gachechiladze2022quantum, our_prl}. For this, we choose a subspace of the NS polytope spanned by a convex combination of two classes of probability distributions, namely the extremal points of the NS polytope, also known as Popescu-Rohrlich (PR) boxes, and the vertices of the local polytope \cite{popescu1994quantum,brunner2014bell}. The probability distribution for each such point is given by 
	\begin{align}
    \begin{split}
		\text{P}^{\text{PR}}_{ijk}(A,B|\alpha,\beta)&=\frac{1}{d}\delta_{A\oplus B,\alpha\cdot\beta\oplus i\cdot\alpha\oplus j\cdot\beta\oplus k },\cr
		\text{P}^{\text{L}}_{ijkl}(A,B|\alpha,\beta)
		&=\delta_{A,i\cdot \alpha\oplus j}\delta_{B,k\cdot \beta\oplus l},
          \end{split}
	\end{align}
	where the indices $i,j,k,l,A,B,\alpha \in [d]$ and $\beta\in\{0,1\}$. We now choose the following slice containing a PR-box, a symmetrized family of local boxes, and white noise
	\begin{align}\label{eq:slice_dist}
    \mathrm{P}
    =(1-s-t)\mathrm{P}^{\mathrm{PR}}_{000}
    +\frac{s}{d}\sum_{r=0}^{d-1}\mathrm{P}^{\mathrm{L}}_{2r0\overline{r}}
    +t\,\mathrm{P}^{\mathrm{wn}},
    \end{align}
    where $\mathrm{P}^{\mathrm{wn}}(A,B\mid\alpha,\beta)=1/d^2$,  $s,t\geq0$, $s+t\leq1$, and $\overline{r}\equiv -r \pmod d$ denotes the additive inverse of $r$. For Alice's input distribution, we choose the following correlated family:
    \begin{equation}\label{eq:dft_dist}
        q_{kl}=\epsilon\,\operatorname{Re}\!\left[\frac{1}{d} +\exp\!\left(\frac{2\pi \mathrm{i}(k+1)(l+1)}{d+1}\right)  \right],
    \end{equation}
where $\{q_{kl}\}_{k,l=0}^{d-1}$ are the biases defined in \cref{eq:input_dist}, and $\epsilon$ controls the correlation strength. We choose a specific channel by setting the biases to
	\begin{align}
		e^i_c=\begin{cases} 
			\frac{-1}{d-1}, & i< \frac{d-1}{2} \\
			\frac{1}{d-1}, & i> \frac{d-1}{2} \\
			0, & \text{otherwise.} 
		\end{cases}
	\end{align}
	With this choice of input distribution and channel parameters, we compute the corresponding NS-box biases $\{e^k_{ij}:i,k\in[d],\,j\in\{0,1\}\}$ and evaluate \cref{eq:d2dd_corr_ineq} on the correlation slice defined above.
\begin{figure}
    \centering

\includegraphics{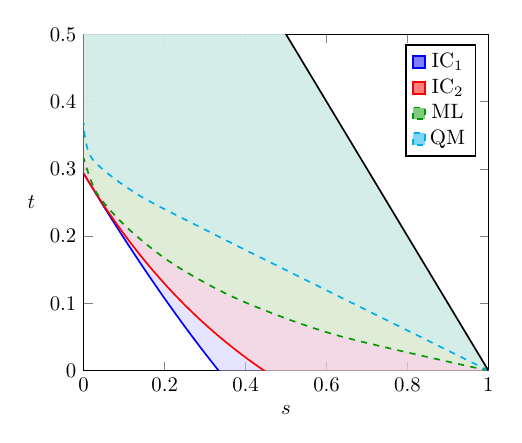}
		\caption{ A comparison of the inequalities obtained from the original formulation of IC ($\text{IC}_1$) and its extension ($\text{IC}_2$) in a chosen correlation slice in the $d2dd$ scenario ($d=3$). The axes correspond to the weights of each distribution described in \cref{eq:slice_dist}. The dotted lines corresponding to ML and quantum mechanics (QM) were computed numerically using the NPA hierarchy.}
		\label{fig:comparison}
	\end{figure}

    We optimize over $\epsilon$ and plot the envelope of the resulting family of inequalities in \cref{fig:comparison} for $d=3$. We observe that the correlated inequality is strictly stronger than the original IC inequality from \cite{gachechiladze2022quantum}. Interestingly, in contrast to the binary case described by \cref{eq:better_Uffink}, the strongest inequality is not obtained in the limit of nearly perfect input correlations. Moreover, for a fixed correlation strength, the corresponding inequality can be either tighter or weaker in different regions of the slice. This behavior persists in higher dimensions; see \cref{app:slice_details} for details.

    Comparing the inequalities obtained from the extended IC principle with those obtained from other principles, we find that, in this particular slice, the extended-IC inequalities derived above are weaker than ML or, equivalently, than the first level $Q_1$ of the NPA hierarchy~\cite{navascues2007bounding}. This should not be interpreted as an inclusion relation between IC and ML, since the inequalities above were derived in the zero-capacity limit of the noisy communication channel. At nonzero capacity, IC can provide bounds that are tighter than those implied by ML; see, for instance, \cite{cavalcanti2010macroscopically, miklin2021information}. However, the algorithm introduced in \cite{our_prl} is not suited to deriving polynomial inequalities in this regime. Instead, one can derive bounds on the biases of individual NS-boxes, or equivalently on the success probability of simulating a PR box. In the next section, we follow this approach and derive bounds on the success probabilities of random access codes.

\section{Bounds on Entanglement-assisted Random Access Codes }

The original IC scenario is an instance of a communication primitive known as a \textit{random access code} (RAC), in which Alice holds a random string and Bob must guess one specified part of it. We denote a RAC by the tuple $(n^{(d)},m,p)$, where $n$ is the number of dits held by Alice, $d$ is their alphabet size, $m<n$ is the number of communicated dits, and $p$ is the minimum probability with which Bob correctly guesses the requested dit; see \cite{pawowskiearac,ambainis2009quantumrandomaccesscodes} for an overview. As in many communication tasks, quantum resources can provide an advantage over classical RAC protocols \cite{Ambainis2024}. Two major models of quantum-assisted RACs are quantum random access codes (QRACs) and entanglement-assisted random access codes (EARACs): in QRACs, Alice communicates qudits, whereas in EARACs the communication remains classical but the parties may share an entangled state.

Considerable work has been devoted to determining the optimal winning probabilities achievable with quantum resources in both models \cite{pawowskiearac,ambainis2009quantumrandomaccesscodes,deVicente2019}. At first sight, teleportation might suggest that the two models are interconvertible, since a protocol involving the transmission of qubits can be simulated by sending classical bits in the presence of shared entanglement. However, the relation between the achievable winning probabilities in the two models is more subtle. For binary outcomes and messages, EARACs outperform QRACs \cite{bin_earac}, and in the $(n^{(2)},1,p)$ scenario every QRAC has an equivalent EARAC \cite{rac_equiv}. In general, however, the two models are inequivalent: EARACs can outperform QRACs assisted only by classical shared randomness in $(n^{(2)},1,p)$ scenarios, whereas in some higher-dimensional cases QRACs outperform EARACs \cite{rac_inequiv,tavakoli_spatial}.

A natural question is how IC restricts RAC winning probabilities. In this section, we derive new analytic bounds on the success probabilities in $(n^{(d)},1,p)$ EARAC scenarios. We consider the ``isotropic'' or ``unbiased errors'' case introduced in \cite{miklin2021information}, where the guessing probabilities for all dits are identical. 
We consider the effect of correlations among Alice's dits on the success probabilities and prove the surprising result that inequalities that consider correlations among Alice's inputs, despite being tighter in some instances, do not lead to better bounds. Building on this observation, we argue for the optimality of the known IC-implied bounds on EARAC winning probabilities \cite{miklin2021information}.

\subsection{Bounds for $(n^{(d)},1,p)$ scenario}
We begin by considering the $(n^{(2)},1,p)$ EARAC scenario. For the ``unbiased errors'' case, we define the guessing probability in the case of a perfect channel as
\begin{equation}
		p_\mathrm{guess}=\mathrm{P}(g=a_i|b=i)=\frac{1+e}{2}, \;\;\forall i\in[n],
	\end{equation}
	where $e\in[-1,1]$ denotes the bias associated with the probability of correctly guessing $a_i$. If the communication channel is noisy, sending the message unchanged with probability $p_c$ and flipping it with probability $1-p_c$, then the overall winning probability is
	\begin{align}
		p=p_\mathrm{guess}\; p_c+(1-p_\mathrm{guess})(1-p_c)=\frac{1+e_ce}{2},
	\end{align}
	where $e_c=2p_c-1$ is the bias associated with the channel's probability $p_c$. We now use IC to upper-bound this winning probability. By Fano's inequality,
	\begin{align}
		n\left(1-h(p)\right)\leq\sum_{i=0}^{n-1}I(g;a_i|b=i)&\leq\mathcal{C}=1-h(p_c)\cr
		\Rightarrow n\left(1-h\left(\frac{1+e_ce}{2}\right)\right)&\leq 1-h\left(\frac{1+e_c}{2}\right),\hspace{1cm}
	\end{align}
	where $h(x)=-x\log_2 x-(1-x)\log_2(1-x)$ is the binary Shannon entropy. We observe that both the left- and right-hand sides of the inequality vanish in the limit $e_c\rightarrow0$. Thus, by using the algorithm as described in \cite{our_prl}, we get a bound on the bias and thus, the associated winning probability in the case of perfect transmission is
	\begin{equation}
		e^2\leq\frac{1}{n}\;\Rightarrow\; p\leq\frac{1}{2}\left(1+\frac{1}{\sqrt n}\right).
	\end{equation}
	Thus, we have re-derived the well known bound on $(n^{(2)},1,p)$ EARAC winning probability \cite{ambainis2009quantumrandomaccesscodes,pawowskiearac}. The same argument extends naturally to dits. We consider the \textit{uniform-error channel}, which sends the message unchanged with probability $p_c$ and maps it to each of the other $d-1$ symbols with probability $(1-p_c)/(d-1)$. In this case, the winning probability is
    \begin{equation}
			p=\frac{1}{n}\sum_{i=0}^{n-1}\text{P}(g=a_i|b=i)=\frac{1+(d-1)e_ce}{d}.
	\end{equation}
    where $e_c=\frac{d p_c-1}{d-1}$ is the bias associated to the probability $p_c$. Using a higher-dimensional version of Fano's inequality and following the same steps, we obtain the following bound on the winning probability in the $(n^{(d)},1,p)$ scenario:
	\begin{equation}\label{eq:earac_general}
		p\leq\frac{1}{d}\left(1+\frac{d-1 }{\sqrt n}\right).
	\end{equation}
    See \cref{app:rac_bounds} for details. 

    Curiously, this is exactly the bound for $(n^{(d)},1,p)$ QRAC scenarios that was recently derived in \cite{Farkas2025simplegeneralbounds}. It is not known in general whether QRAC bounds are tighter than EARAC bounds, although numerical evidence suggests that this may be the case in some regimes \cite{tavakoli_spatial}. Nevertheless, the bound in \cref{eq:earac_general} is known not to be tight for QRACs \cite{Ambainis2024}. This leaves open two possibilities: either the same bound applies to both models, in which case EARACs are no more powerful than QRACs for this class of tasks, or EARACs can outperform QRACs in higher dimensions. Standard derivations of such bounds typically use operator-theoretic methods, assuming a Hilbert-space structure and optimizing suitable operator norms over measurements and states. In contrast, the derivation above is theory-independent: it relies only on IC, is conceptually minimal in its assumptions, and is technically simpler.

	\subsection{Optimality of EARAC bounds}\label{sec:optimality}
	One of the early applications of IC showed that some correlations satisfying Macroscopic Locality nevertheless violate IC \cite{cavalcanti2010macroscopically}. This was demonstrated using an isotropic family of PR boxes whose biases are defined using the indicator function,  $[\cdot]$:
    \begin{equation}\label{eq:iso_pr}
        e^k_{ij}=[k= i\cdot j]e-[k\neq i\cdot j]\frac{e}{d-1},
    \end{equation}
    where $i,k\in [d]$ and $j\in\{0,1\}$. The main idea was to mix extremal non-signaling distributions with white noise and study the critical noise levels above which a violation of IC occurs. In more recent work, these critical-noise bounds were improved by using a noisy channel with a single box instead of concatenating several boxes \cite{miklin2021information}.

    The bias associated with the winning probability can equivalently be interpreted as the complement of the amount of white noise mixed into a perfect PR box. Indeed, for the mixture
$
\mathrm{P}_{\mathrm{mix}}=e\,\mathrm{P}_{\mathrm{PR}}+(1-e)\mathrm{P}_{\mathrm{wn}}$, $
\mathrm{P}_{\mathrm{wn}}(A,B\mid\alpha,\beta)=\frac{1}{d^2},
$
the associated bias is precisely $e$, since white noise has zero bias. We list several known bounds from the literature \cite{tavakoli_spatial,miklin2021information} in \cref{tab:table1}.

	\begin{table}[h]
		\caption{\label{tab:table1}%
			A collection of RAC bounds from the literature. Here, $(n,d)$ denotes the number of dits Alice receives and their alphabet size, respectively. The first three columns list the maximum winning probabilities achievable with classical shared randomness ($p^C$), quantum resources ($p^E_{\mathrm{num}}$), and almost-quantum correlations ($p^E_{Q_{1+AB}}$), respectively. The final column lists the optimal winning probabilities subject to the IC principle ($p^{IC}$), obtained numerically. The analytical bounds in \cref{eq:earac_general}, derived in the zero-capacity limit, are looser than the numerical IC bounds shown here, except for $d=2$.
		}
		\begin{ruledtabular}
			\begin{tabular}{lrrrr}
				\textrm{(n,d)}& \textrm{$p^C$}& \textrm{$p^E_{num}$} & \textrm{$p^E_{Q_{1+AB}}$} & \textrm{$p^{IC}$}\\ 
				\colrule
				(2,2) & 0.7500 & 0.8536 & 0.8536 & 0.8536 \\
				(2,3) & 0.6667 & 0.7778 & 0.7778 & 0.8016 \\
				(2,4) & 0.6250 & 0.7405 & 0.7471 & 0.7718 \\
				(2,5) & 0.6000 & 0.7178 & 0.7179 & 0.7517 \\
				(3,3) & 0.6296 & 0.6854 & 0.6912 & 0.7142 \\
			\end{tabular}
		\end{ruledtabular}
	\end{table}
	
	We emphasize that the analytical bound in \cref{eq:earac_general} is not the tightest bound derivable from IC. In particular, by numerically optimizing the winning probability over the input distribution and the communication channel, one obtains the stronger bounds listed in \cref{tab:table1}. In general, the optimal bound occurs at nonzero capacity. Interestingly, these bounds exactly match those obtained from the original IC statement in \cite{miklin2021information}.

    Since the extended IC statement can yield tighter inequalities, it is natural to ask whether it also improves the bounds on EARAC winning probabilities. Surprisingly, for the isotropic NS boxes defined in \cref{eq:iso_pr}, the answer is negative. We show that introducing correlations among Alice's inputs does not increase the left-hand side of the extended IC inequality and therefore cannot lead to stronger bounds than those obtained from the original IC principle in \cite{miklin2021information}. This suggests that the latter bounds are optimal within the framework considered here. The argument has two parts, which we present below. For this purpose, let us define the left-hand side of the extended IC as the \textit{total information}
    \begin{equation}
      I_\epsilon(e,\{e_c^i\}_{i=0}^{d-1}):=I(g;a_0|b=0)+I(g;a_1|b=1,a_0).
     \end{equation}

    \begin{theorem}\label{thm:corr_weak}
    In the $d2dd$ scenario, consider isotropic NS boxes and any symmetric communication channel. For any admissible correlated input distribution of Alice's dits, the total information appearing on the left-hand side of the extended IC statement is upper-bounded by its value for independent, uniformly distributed inputs. Equivalently,
    \begin{equation}
    I_\epsilon(e,\{e_c^i\}_{i=0}^{d-1})
    \leq
    I_0(e,\{e_c^i\}_{i=0}^{d-1}),
    \end{equation}
    where $I_0$ denotes the case $\epsilon=0$.
    \end{theorem}

	The proof is given in \cref{app:optimality_proof1}. Second, we show that among all symmetric channels of the form in \cref{eq:input_dist}, the uniform-error channel yields the optimal bounds on the winning probabilities. 
	
    \begin{theorem}\label{thm:opt_channel}
		In the $d2dd$ scenario, consider isotropic NS boxes and uniformly distributed Alice inputs. Among all symmetric communication channels of the form in \cref{eq:input_dist}, the tightest bound on the isotropic bias, and hence on the winning probability, is achieved by the uniform-error channel. Equivalently, the optimum is attained when all channel biases except one are identical, i.e., $\{e_c^i\}_{i=0}^{d-1}=\{r,\dots,r,-(d-1)r\}$ up to permutations.
	\end{theorem}
    
	We refer to \cref{app:optimality_proof2} for a proof of the above theorem. From the above two theorems, we see that the uniform distribution and the uniform error channel are optimal, which lets us conclude the following:
	\begin{corollary}\label{conj:2}
        In the $d2dd$ scenario, the previously known IC-derived EARAC bounds remain optimal under the extended formulation of information causality when optimizing over arbitrary input distributions and arbitrary symmetric communication channels.
	\end{corollary}

    \section{Discussions}
	
	In this work, we studied the consequences of the extended Information Causality principle for quantum Bell inequalities and random access codes. Our results show that the extended formulation can lead to strictly stronger constraints than the original IC principle, but also that this additional strength is highly scenario-dependent.
    
    First, we derived a family of analytical quantum Bell inequalities in the $nn22$ scenario. These inequalities strengthen the previously known IC-derived inequalities and lead to an improved analytical upper bound on the Collins--Gisin $I_{nn22}$ family. A noteworthy feature of this result is that the improvement already appears for protocols in which Alice's inputs are independent. Thus, the extended IC statement can yield stronger constraints even when the protocol does not explicitly exploit correlations among Alice's inputs.
    
    We then considered the $d2dd$ scenario. In contrast to the $nn22$ case, the resulting inequalities depend not only on the biases of the non-signaling box, but also on the choice of Alice's input distribution and on the parameters of the classical communication channel. By choosing correlated input distributions, we obtained inequalities that improve upon the previously known IC-derived quantum Bell inequalities. However, in the correlation slice considered in this work, these inequalities remain weaker than Macroscopic Locality, or equivalently the first level $Q_1$ of the NPA hierarchy. Since our polynomial inequalities are derived in the zero-capacity regime, this comparison should not be interpreted as a general relation between IC and ML (especially, when it is known that macroscopically local correlations violate the original IC statement). Rather, it points to a limitation of the current method for extracting tractable polynomial constraints from IC.

    Finally, we applied IC to entanglement-assisted random access codes. We recovered the known IC bound for binary EARACs and derived analytical bounds for higher-dimensional $(n^{(d)},1,p)$ scenarios. These bounds were obtained in a theory-independent manner, relying only on information-theoretic constraints rather than on the Hilbert-space formalism of quantum theory. We also analyzed whether the extended IC statement can improve the known IC-derived EARAC bounds. For isotropic non-signaling boxes in the $d2dd$ scenario, we found that it cannot: introducing correlations among Alice's inputs does not increase the relevant total information term, and the optimal bound is achieved by the uniform input distribution together with the uniform-error channel. Therefore, within the class of protocols and channels considered here, the previously known IC-derived EARAC bounds remain optimal under the extended formulation of Information Causality.

    Taken together, these results show that the role of the extended IC inequality is subtle. In several Bell-inequality scenarios, the extended IC inequality in \cref{eq:ic_statement_correlated} leads to stronger constraints than the original IC statement, even when no correlations are explicitly introduced in Alice's input distribution. This is also supported by recent work~\cite{pollyceno2026communication}, where the authors re-derived the same extended IC inequality (Eq.(10) in~\cite{pollyceno2026communication}) using Fourier-Motzkin elimination and applied it to the $(4,4;2,2)$ Bell scenario. They showed that this inequality yields better bounds than the original IC inequality, again without relying on correlated inputs. Thus, while correlations among Alice's inputs can be useful in some settings (e.g., in the $d2dd$ scenario), the improvement provided by the extended IC inequality is not solely due to the presence of such correlations. Understanding precisely when correlations themselves provide an additional advantage remains an interesting open question.

    Several directions remain open. First, in the $d2dd$ scenario, the correlated input distribution and communication channel were chosen in a structured but not fully optimized way. It would be interesting to characterize the input distributions and channels that produce the tightest extended-IC inequalities. Second, the present method extracts polynomial inequalities by taking the zero-capacity limit. Since IC can give stronger constraints at nonzero capacity, developing methods for deriving useful analytical or polynomial constraints away from this limit would be valuable. Third, on the random-access-code side, it would be natural to extend the analysis to more general $(n^{(d)},m,p)$ EARACs and to protocols beyond the ones considered here. Such extensions may clarify whether stronger theory-independent bounds can be obtained from IC in more general communication settings.

	\begin{acknowledgments}
			We thank Lucas Vieira and Jan Nöller for fruitful discussions. This research was funded by the Deutsche Forschungsgemeinschaft (DFG, German Research Foundation), project number 441423094, the National Research Center for Applied Cybersecurity ATHENE. We acknowledge funding from the QuantERA II Programme, that has received funding from the EU’s H2020 research and innovation programme under the GA No 101017733.
	\end{acknowledgments}

    \section*{Data Availability}
    All data and code associated with the symbolic and numeric computations in this paper are available from a GitHub repository \cite{github_code}.

\onecolumngrid
\appendix
\section*{APPENDIX}
    This appendix contains additional information and details of the theorems in the main text.
\section{Details of the $nn22$ scenario}\label{app:IC_nn22}
In this section we provide proofs of \cref{thm:nn22_ineq} and \cref{thm:Inn22} which we reproduce for convenience below.
\subsection{Details on $nn22$ inequalities from the extended IC principle}\label{app:nn22_ineqs}  
\thmnn*
\begin{proof}
We consider a general class of encoding-decoding protocols for the scenario with $n$ input bits for Alice. In particular, we take the following,	\begin{align}\label{eq:app_nn22_protocol}
		\alpha=f(\vec{a}), \quad \beta=b, \quad x=h(\vec{a})\oplus A, \quad 
		g=x'\oplus B,
	\end{align}
where $f: \{0,1\}^n\to [n]$ and $h: \{0,1\}^n\to \{0,1\}$ are discrete-valued function of the inputs bits $\{a_i\}_{i=0}^{n-1}$, which we arranged as a binary vector $\vec{a}\in \{0,1\}^n$. We choose the communication of Alice to Bob to be over a binary symmetric channel with the probability of transmitting a correct message being $\frac{1+e_c}{2}$, i.e., $\Pr(x'=0 \mid x=0) = \Pr(x'=1 \mid x=1) = \frac{1+e_c}{2}$. Note that there are still other classes of protocols that one can consider for the $nn22$ scenario. However, we leave the analysis of such protocols for future investigations.

We denote by $\mathbb{A}_i=\{a_0,\dots,a_{i-1}\}$ the set of first $i$ input bits. The extended IC statement is given by
\begin{equation}\label{eq:app_nn22_ic}
    \sum_{i=0}^{n-1}I(g;a_i|\mathbb{A}_i,b=i)\leq \mathcal{C}.
\end{equation}
We expand each term in the summation as follows
\begin{align}
I(g;a_i|\mathbb{A}_i,b=i)&=H(g|\mathbb{A}_i,b=i)-H(g|\mathbb{A}_{i+1},b=i)\cr
&=\sum_{\vec{k}\in\{0,1\}^i} \mathrm P (\mathbb{A}_i=\vec{k})H(g|\mathbb{A}_i=\vec{k},b=i)-\sum_{\vec{k}\in\{0,1\}^{i+1}} \mathrm P(\mathbb{A}_{i+1}=\vec{k})H(g|\mathbb{A}_{i+1}=\vec{k},b=i).
\end{align}
To compute the relevant entropies, we write down the joint distribution of Alice's inputs and Bob's guess as follows

\begin{align}\label{eq:app:nn22_joint}
\mathrm{P}(g=j,\mathbb{A}_n=\vec{k}\mid b=i)
&=\mathrm{P}(\mathbb{A}_n=\vec{k})\,
  \mathrm{P}(g=j\mid\mathbb{A}_n=\vec{k},b=i)\nonumber\\
&=\mathrm{P}(\mathbb{A}_n=\vec{k})
  \sum_{m=0}^{n-1}
  \mathrm{P}(\alpha=m\mid\mathbb{A}_n=\vec{k})\,
  \mathrm{P}(g=j\mid\alpha=m,\beta=i)\nonumber\\
&=\mathrm{P}(\mathbb{A}_n=\vec{k})
  \sum_{m=0}^{n-1}
  \delta_{m,f(\vec{k})}\,
  \mathrm{P}\!\left(
  N\oplus(A\oplus B)=h(\vec{k})\oplus j
  \mid \alpha=m,\beta=i
  \right)\nonumber\\
&=\mathrm{P}(\mathbb{A}_n=\vec{k})\,
  \mathrm{P}\!\left(
  N\oplus(A\oplus B)=h(\vec{k})\oplus j
  \mid \alpha=f(\vec{k}),\beta=i
  \right)\nonumber\\
&=\frac{1}{2^n}
  \frac{1+(-1)^{h(\vec{k})\oplus j}e_c e_{f(\vec{k}),i}}{2}, 
\end{align}
Here $N$ is the binary channel-noise variable, defined by $
x'=x\oplus N$, 
$\mathrm{P}(N=0)=\frac{1+e_c}{2}$ and 
$\mathrm{P}(N=1)=\frac{1-e_c}{2}$.
Using the joint distribution, we can then calculate the above two marginals to be
\begin{equation}
p(g=j|\mathbb{A}_i=\vec{k},b=i)=\frac{\sum_{k_{j\geq i}}p(g=j,\mathbb{A}_n=\vec{k}|b=i)}{p(\mathbb{A}_i=\vec{k})}=\frac{1}{2^{n-i}}\sum_{k_{j\geq i}}\frac{1+(-1)^{h(\vec{k})\oplus j}e_c e_{f(\vec{k}),i}}{2},
\end{equation}

\begin{equation}
\mathrm P(g=r|\mathbb{A}_{i+1}=\vec{k},b=i)=\frac{\sum_{k_{j\geq i+1}}\mathrm P(g=r,\mathbb{A}_n=\vec{k}|b=i)}{\mathrm P(\mathbb{A}_{i+1}=\vec{k})}=\frac{1}{2^{n-i-1}}\sum_{k_{j\geq i+1}}\frac{1+(-1)^{h(\vec{k})\oplus r}e_c e_{f(\vec{k}),i}}{2}.
\end{equation}
We define the auxiliary variables $c_i,d_i$ as follows
\begin{equation}
    c_i:=\frac{1}{2^{n-i}}\sum_{k_{j\geq i}}(-1)^{h(\vec{k})}e_{f(\vec{k}),i},\quad  d_i:=\frac{1}{2^{n-i-1}}\sum_{k_{j\geq i+1}}(-1)^{h(\vec{k})}e_{f(\vec{k}),i}
\end{equation}
Using that, the expression for the marginal distributions is simply
\begin{align}
    \mathrm P (g=j|\mathbb{A}_i=\vec{k},b=i)=\frac{1+(-1)^je_cc_i}{2}
    \\
    \mathrm P (g=j|\mathbb{A}_{i+1}=\vec{k},b=i)=\frac{1+(-1)^je_cd_i}{2}
\end{align}
We now divide inequality in \cref{eq:app_nn22_ic} by the capacity on the right hand side and obtain
\begin{equation}
    \sum_{i=0}^{n-1}I_i\leq 1,\quad \text{where } I_i:=\frac{I(g;a_i|\mathbb{A}_i,b=i)}{1-H\left(\frac{1+e_c}{2}\right)}
\end{equation}
Plugging the computed entropies back in the IC statement we have
\begin{align}
    I_i&=\frac{1}{1-H\left(\frac{1+e_c}{2}\right)}\left(\sum_{\vec{k}\in\{0,1\}^i} \mathrm P(\mathbb{A}_i=\vec{k})H(g|\mathbb{A}_i=\vec{k},b=i)-\sum_{\vec{k}\in\{0,1\}^{i+1}} \mathrm P(\mathbb{A}_{i+1}=\vec{k})H(g|\mathbb{A}_{i+1}=\vec{k},b=i)\right)\cr
    &=\frac{1}{1-H\left(\frac{1+e_c}{2}\right)}\left(\sum_{k_{j<i}}\frac{1}{2^{i}}H\left(\frac{1+e_cc_i}{2}\right)-\sum_{k_{j<i+1}}\frac{1}{2^{i+1}}H\left(\frac{1+e_cd_i}{2}\right)\right).
\end{align}
Next, we adhere to the usual method of computing the limit of each $I_i$, when the channel capacity parameter, $e_c\rightarrow 0$, and we denote it by $\tilde{I}_i$,
\begin{align}
   \tilde I_i
   &=
   -\sum_{k_{j<i}}\frac{c_i^2}{2^{i}}
   +\sum_{k_{j<i+1}}\frac{d_i^2}{2^{i+1}} \nonumber\\
   &=
   \frac{1}{2^{2n-i}}
   \sum_{k_{j<i}}
   \left(
   2\sum_{k_i}
   \left(\sum_{k_{j> i}}(-1)^{h(\vec{k})}e_{f(\vec{k}),i}\right)^2
   -
   \left(\sum_{k_{j\geq i}}(-1)^{h(\vec{k})}e_{f(\vec{k}),i}\right)^2
   \right) \nonumber\\
   &=
   \frac{1}{2^{2n-i}}
   \sum_{k_{j<i}}
   \left(
   \sum_{k_i}(-1)^{k_i}
   \sum_{k_{j> i}}(-1)^{h(\vec{k})}e_{f(\vec{k}),i}
   \right)^2 .
\end{align}
Summing over $i$ we finally have the general form for the extended IC inequality for the $nn22$ scenario:
\begin{equation}\label{eq:gen_nn22_ineq}
     \sum_{i=0}^{n-1}2^{i}\sum_{k_{j<i}}\left(\sum_{k_i}(-1)^{k_i}\sum_{k_{j> i}}(-1)^{h(\vec{k})}e_{f(\vec{k}),i}\right)^2\leq 4^n.
  \end{equation}
This finishes the proof of the general family of inequalities.  
\end{proof}

Next, we will prove a specific inequality which will be helpful later in proving the bound on $I_{nn22}$ inequalities in \cref{app:inn22_details}.
\begin{theorem}\label{thm:special_ineq}
    In the $nn22$ Bell scenario, the following inequality follows from the Information Causality principle,  
     \begin{equation}\label{eq:app_res_dd22_relevant}
     \left(2e_{0,0}+\sum_{j=1}^{n-1}2^{j}e_{j,0}\right)^2+\sum_{i=1}^{n-1}2^{i-1}\left(2e_{0,i}+\sum_{j=1}^{n-i}(-1)^{\delta_{j,n-i}}2^{j}e_{j,i}\right)^2\leq 4^{n},
  \end{equation}
where $e_{j,i}\coloneqq 2\Pr(A=B \mid \alpha=j,\beta=i)-1$, for all $i,j\in \{0,1,\dots,n-1\}$, and $e_{n,0}\coloneqq 0$.
    \end{theorem}

\begin{proof}
We now choose a specific protocol where Alice chooses the following encoding for her input and message

\begin{equation}
     f(\vec{a})=n-1-\sum_{i=1}^{n-1}\prod_{j=1}^i(a_0\oplus a_j),\quad  h(\vec{a})=a_0.
\end{equation}

Plugging in the protocol functions in the general inequality in \cref{eq:gen_nn22_ineq} we have
\begin{equation}\label{eq:app_nn22_specifc_ineq}
    \tilde I_i=\frac{1}{2^{2n-i}}\sum_{k_{j<i}}\left(\sum_{k_i}(-1)^{k_i}\sum_{k_{j> i}}(-1)^{k_0}e_{f(\vec{k}),i}\right)^2.
\end{equation}

To tackle the above sum, we first understand the structure of $f$ more carefully. We define $P_i:= \prod_{j=1}^i (a_0 \oplus a_j)$ and immediately notice that $P_i=0$ iff there exists $j$ such that $a_j=a_0$ and $P_i=1$ otherwise. We also notice that $P_i=0\implies P_{i+1}=0$. Thus, if $1\leq m\leq n-1$ denotes the smallest index such that $a_m=a_0$ then $f(\vec{a})=n-1-\sum_{i=1}^{n-1}P_i=n-1-(m-1)=n-m$. If no such $m$ exists, then we have simply $f(\vec{a})=0$. Hence, whenever  the bit $a_{n-r}\neq a_0$ we have $f(\vec{a})=r$.
\\

We treat the case of $i=0$ separately. We notice that in the definition of $f$, the function takes a value `$0$' twice (for strings $\vec{a}=01\dots1,10\dots0$) and similarly it takes the value `$0<j\leq n-1$' when $a_0=a_1=\dots a_{n-j-1}\neq a_{n-j}$. This fixes the first $n-j+1$ bits and leaves $j-1$ bits free. Thus $f(\vec{a})=j$ occurs $2^j$ times, 
\begin{align}
    \tilde I_0&=\frac{1}{2^{2n}}\left(\sum_{k_0}(-1)^{k_0}\sum_{k_{j> 0}}(-1)^{k_0}e_{f(\vec{k}),0}\right)^2=\frac{1}{2^{2n}}\left(\sum_{\vec{k}}e_{f(\vec{k}),0}\right)^2=\frac{1}{2^{2n}}\left(2e_{0,0}+\sum_{j=1}^{n-1}2^{j}e_{j,0}\right)^2.
\end{align}
For the case when $i>0$, we first consider the inner sum in \cref{eq:app_nn22_specifc_ineq}. We notice that $k_0$ is independent of the sums and will vanish when the expression is squared, so we drop it in the inner sum. Let $k_0=c$, now since $\{k_j\}_{j<i}$ are fixed, if any $k_j=k_0$ for $1\leq j<i$, then $f(\vec{k})$ is constant over $k_i$ and $k_{j>i}$, and the factor $\sum_{k_i}(-1)^{k_i}=0$ forces $\tilde I_i=0$. Thus, in the outer sum, the only configuration which survives is of the form $k_1=\dots=k_{i-1}\neq k_0$.

The smallest index $m$ such that $f(\vec{k})=n-m$ has three possibilities. Correspondingly, we split the expression into three cases.

\begin{itemize}
    \item \textbf{Case 1:  $m<i$.} In this case $f(\vec{k})$ is independent of $\{k_j\}_{j>i}$ and thus the inner sum collapses since
    \begin{equation}
        \sum_{k_i}(-1)^{k_i}\sum_{k_{j> i}}e_{f(\vec{k}),i}=\left(\sum_{k_i}(-1)^{k_i}\right)\left(\sum_{k_{j> i}}e_{n-m,i}\right)=0.
    \end{equation}
    \item \textbf{Case 2: $m=i$.} In this case $k_i=k_0=c$ and thus $f(\vec{k})=n-i$. Hence, the inner sum is simply
    \begin{equation}
        \sum_{k_i}(-1)^{k_i}\sum_{k_{j> i}}e_{f(\vec{k}),i}=(-1)^c2^{n-i-1}e_{n-i,i}.
    \end{equation}
    \item \textbf{Case 3: $m>i$.} In this case, $k_i=c\oplus 1$ otherwise we would have $m=i$. Now we split the summation upto the index $m$ and then the rest of the terms 
        \begin{equation}
        \sum_{k_i}(-1)^{k_i}\left(\sum_{k_{j<m}}\sum_{k_{j\geq m}}e_{f(\vec{k}),i}\right)=\sum_{k_i}(-1)^{k_i}\left(\sum_{k_{j<m}}2^{n-m}e_{n-m,i}\right)=(-1)^{c\oplus 1}\left(\sum_{r=1}^{n-1-i}2^{r-1}e_{r,i}+e_{0,i}\right).
    \end{equation}
Adding all these contributions we have the inner sum to be equal to
    \begin{equation}
        \frac{(-1)^{c\oplus 1}}{2}\left(2e_{0,i}+\sum_{j=1}^{n-i-1}2^j e_{j,i}-2^{n-i}e_{n-i,i}\right).
    \end{equation}
\end{itemize}
So squaring and adding the two contributions we finally have
\begin{equation}
    \tilde I_i=\frac{1}{2^{2n-i}}\sum_{c=0}^1\frac{1}{4}\left(2e_{0,i}+\sum_{j=1}^{n-i-1}2^j e_{j,i}-2^{n-i}e_{n-i,i}\right)^2=\frac{1}{2^{2n-i+1}}\left(2e_{0,i}+\sum_{j=1}^{n-i}(-1)^{\delta_{j,n-i}}2^{j}e_{j,i}\right)^2.
\end{equation}

Thus, we finally have the inequality
\begin{equation}
    \left(2e_{0,0}+\sum_{j=1}^{n-1}2^{j}e_{j,0}\right)^2+\sum_{i=1}^{n-1}2^{i-1}\left(2e_{0,i}+\sum_{j=1}^{n-i}(-1)^{\delta_{j,n-i}}2^{j}e_{j,i}\right)^2\leq 4^{n}.
\end{equation}
\end{proof}

\subsection{Details of $I_{nn22}$ upper bound from the extended IC statement }\label{app:inn22_details}
In this section, we derive an improved analytical upper bound on the quantum value of the $I_{nn22}$ expression~\cite{collins2004relevant} for an arbitrary $n$, using the inequalities in \cref{thm:special_ineq}. 
We consider the case of NS-boxes with uniformly distributed marginals, in which case the $I_{nn22}$ Bell inequality takes the form,
\begin{equation}\label{eq:app_Inn22}
    I_{nn22}\coloneqq \frac{-n^2+n-2}{8}+\frac{1}{4}\left(\sum_{i=0}^{n-1}\sum_{j=0}^{n-i-1}e_{j,i}-\sum_{i=1}^{n-1}e_{n-i,i}\right)\leq 0.\end{equation}
For such family of boxes, we prove the following result.
\thmInn*
\begin{proof}

To arrive at the above result, we bound the $I_{nn22}$ expression in \cref{eq:app_Inn22} subject to the inequality derived in \cref{thm:special_ineq}. So, we need to solve the following optimization problem,
\begin{equation}\label{eq:app_Inn22_opt_1}
    \begin{aligned}
\max_{e_{j,i}} \quad & \sum_{i=0}^{n-1}\sum_{j=0}^{n-i-1}e_{j,i}-\sum_{i=1}^{n-1}e_{n-i,i},\\
\textrm{s.t.} \quad & \left(2e_{0,0}+\sum_{j=1}^{n-1}2^{j}e_{j,0}\right)^2+\sum_{i=1}^{n-1}2^{i-1}\left(2e_{0,i}+\sum_{j=1}^{n-i}(-1)^{\delta_{j,n-i}}2^{j}e_{j,i}\right)^2\leq 4^{n},\\
  &-1\leq e_{j,i}\leq 1,\quad \forall (j,i)\in [\text{all}],
\end{aligned}
\end{equation}
where
$
[\mathrm{all}] \coloneqq \{(j,0):0\leq j\leq n-1\} \cup \{(j,i):1\leq i\leq n-1,\;0\leq j\leq n-i\} $
is the set of indices of all optimization variables appearing in \cref{eq:app_Inn22_opt_1}. As a first step, we perform a change of variables $e_{j,i}\to -e_{j,i}$ for $i+j=n$, $i\in\{1,2\dots,n-1\}$, bringing the problem in Eq.~\eqref{eq:app_Inn22_opt_1} to, 
\begin{equation}\label{eq:inn22_opt_problem}
    \begin{aligned}
\max_{e_{j,i}} \quad & \sum_{i=0}^{n-1}\sum_{j=0}^{n-i-1}e_{j,i}+\sum_{i=1}^{n-1}e_{n-i,i}\\
\textrm{s.t.} \quad & \left(2e_{0,0}+\sum_{j=1}^{n-1}2^{j}e_{j,0}\right)^2+\sum_{i=1}^{n-1}2^{i-1}\left(2e_{0,i}+\sum_{j=1}^{n-i}2^{j}e_{j,i}\right)^2 \leq 4^{n}\\
  &-1\leq e_{j,i}\leq 1,\quad \forall (j,i)\in [\text{all}]. 
\end{aligned}
\end{equation}
We also introduce the following notation,
\begin{equation}\label{eq:app_inn22_Si}
    S_0\coloneqq 2e_{0,0}+\sum_{j=1}^{n-1}2^{j}e_{j,0},\quad  S_i\coloneqq \left(2e_{0,i}+\sum_{j=1}^{n-i}2^{j}e_{j,i}\right),\; i\in \{1,2,\dots,n-1\}.
\end{equation}

We define a function $\mathcal{B}(n)= n(n+3)/2+2\sqrt{2+2^{4-n}}-7$ and our objective $\mathcal L=\sum_{i=0}^{n-1}\sum_{j=0}^{n-i-1}e_{j,i}+\sum_{i=1}^{n-1}e_{n-i,i}$. Now consider the following decomposition
\begin{equation}
    \mathcal{B}(n)-\mathcal L=\sum_{i=0}^{n-1}\sum_{j=0}^{n-i-1}a_{ji}(1-e_{ji})+\sum_{i=1}^{n-1}b_{i}(1-e_{n-i,i})+\lambda \sum_{i=0}^{n-1}c_i\left(S_i-t_i\right)^2+\lambda(4^n-S_0^2-\sum_{i=1}^{n-1}2^{i-1}S_i^2),
\end{equation}
where the parameters are given by
\begin{align}
    c_0&=1,\quad c_i=2^{i-1} \quad \forall  i\in \{1,\dots,n-1\},\quad x=-1+\sqrt{\frac{1}{2}+\frac{1}{2^{n-2}}}, \cr
    t_0&=t_1=t_2=2^{n-1}(1+x),\quad t_i=2^{n-i+1} \quad  \forall  i\in \{3,\dots,n-1\},\cr
    \lambda&=\frac{1}{2^nt_0},\quad
    a_{00}=1-4\lambda t_0,\quad a_{j0}=1-2^{j+1}\lambda t_0 \quad \forall  j\in \{1,\dots,n-1\},\cr 
    a_{0i}&=1-2^{i+1}\lambda t_i,\quad a_{ji}=1-2^{i+j}\lambda t_i \quad \forall i\in \{1,\dots,n-1\}, j\in \{1,\dots,n-i\},\cr
    b_i&=1-2^n\lambda t_i \quad \forall  i\in \{1,\dots,n-1\}.
\end{align}
The value of these coefficients were inspired by the solution of the dual variables from the stationarity equations of the KKT conditions. The particular value of $x$ is chosen so that the constant term in the decomposition equals $\mathcal B(n)$. Plugging these values in the above expression, one can straightforwardly check that the equality holds as follows: we split the right hand side into constant terms and terms involving the biases. It is easy to check that all the quadratic terms cancel each other and the coefficient of the biases match exactly the left hand side. Therefore, all that remains is to check the constant terms which are given by
\begin{align}
    \lambda\left(t_0^2+\sum_{i=1}^{n-1}2^{i-1}t_i^2+4^n\right)+\sum_{i=0}^{n-1}\sum_{j=0}^{n-i-1}a_{ji}+\sum_{i=1}^{n-1}b_{i}&=\frac{1}{2}\left(n^2+3n-2+2^n\lambda[2^n(2x^2-2x-3)+8]\right)\cr&=n(n+3)/2+2\sqrt{2+2^{4-n}}-7=\mathcal{B}(n).
\end{align}
Further, we can see that on the right hand side, since each parameter is non-negative, every term is manifestly non-negative and thus
\begin{equation}
    \mathcal{B}(n)-\mathcal L\geq 0 \implies I_{nn22}\leq \frac{n-4}{2}+\sqrt{\frac{1}{2}+\frac{1}{2^{n-2}}}.
\end{equation}
\end{proof}

\subsection{On the bounding power of the extended IC statement and the derived inequalities}\label{app:iso_nn22}

In the previous section, we used the quadratic inequalities derived from the extended IC statement to bound the $I_{nn22}$ expression. Alternatively, one could have used the IC statement itself, i.e., the logarithmic inequality, as a constraint to derive a bound. In principle, this may provide better bounds because, when extracting polynomial inequalities from the logarithmic ones, one might lose tightness, since the resulting inequalities are obtained in the zero-capacity limit. From \cite{miklin2021information}, it is known that for the $d2dd$ scenario, the tightest bounds occur at a point of nonzero capacity. On the other hand, in \cite{our_prl}, it was shown that in the $nn22$ scenario, for a wide class of protocols, the quadratic inequalities provide an exact characterization of the IC set, i.e., the derived inequalities are as tight as the original logarithmic constraint.

In this section, rather than tackling the same problem generally for the extended IC statement, we demonstrate that the quadratic inequalities are as good as the original extended IC inequality in an isotropic family. We work in a simplified isotropic version of the $I_{nn22}$ scenario, where we consider a family of PR boxes with biases defined by
\begin{equation}
    e_{ij}=(1-2\delta_{i+j,n})e.
\end{equation}

For this family of boxes, the $I_{nn22}$ expression defined in \cref{eq:app_Inn22} takes the affine form
\begin{equation}\label{eq:ison22}
    I^{\mathrm{iso}}_{nn22}
    =
    \frac{-n^2+n-2}{8}
    +
    \frac{n^2+3n-2}{8}e.
\end{equation}

Thus, to bound the $I^{\mathrm{iso}}_{nn22}$ expression, it suffices to bound the bias $e$. For isotropic boxes, the inequality derived in \cref{thm:special_ineq} simplifies to
\begin{equation}\label{eq:ineqison22_extended}
    (3-2^{2-n})e^2\leq 1.
\end{equation}
Therefore,
\begin{equation}\label{eq:ison22_bound}
     e\leq\frac{1}{\sqrt{3-2^{2-n}}}.
\end{equation}

For large $n$, this bound approaches $1/\sqrt 3$ from above. We now calculate the same bound from the logarithmic inequality in \cref{eq:app_nn22_ic}. That is, for a fixed channel parameter $e_c$, we find the critical value of $e$ at which the IC inequality is saturated. Sweeping over the parameter range $0\leq e_c\leq 1$, we obtain a curve of critical values of $e$. We repeat this for $n=3,\dots,10$ and in each case find that the critical bound occurs in the limit of vanishing $e_c$, with values matching \cref{eq:ison22_bound} up to numerical precision; see \cref{fig:nn22_plot}.

This provides numerical evidence that, in this isotropic family, the constraints derived from the quadratic inequalities are as strong as the logarithmic extended-IC constraints. This is in contrast to the $d2dd$ scenario, where the tightest constraints occur at a point of nonzero channel capacity.

\begin{figure}[h]
    \centering
\includegraphics{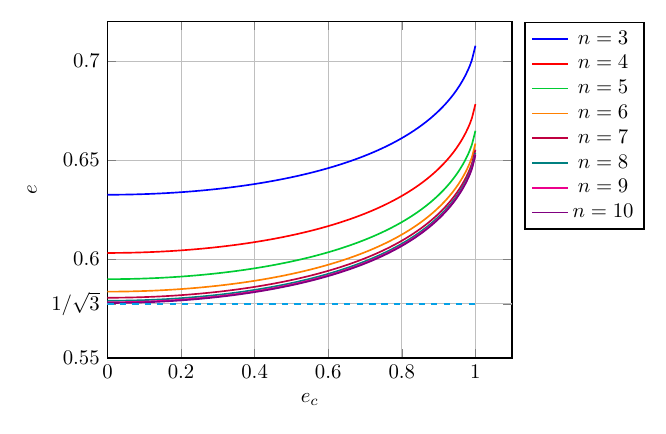}
    \caption{Plot of critical values of $e$ vs. the channel parameter $e_c$ for $n=3,\dots,10$.}
    \label{fig:nn22_plot}
\end{figure}

\section{Details on $d2dd$ inequalities}
\subsection{Details of $d2dd$ inequalities}\label{app:d2dd_details}
		In this section, we provide additional details and the proof of \cref{thm:corr_ineq}. We reproduce below the probability distributions for Alice's inputs, the channel and the NS-box from \cref{eq:input_dist} as follows,
	       \begin{align}
		\mathrm{P}(a_0=k,a_1=l)&=\frac{1+q_{kl}}{d^2}\cr
		\mathrm{P}(x'=x\oplus m)&=\frac{1+(d-1)e_{c}^m}{d}\cr
		\mathrm{P}(A\oplus B=k|\alpha=i,\beta=j)&=\frac{1+(d-1)e^k_{ij}}{d}.
	       \end{align}
        We assume the van Dam protocol generalised for dits i.e.
           \begin{equation}
               \alpha=\overline{a_0}\oplus a_1,\quad\beta=b,\quad x=a_0\oplus A,\quad g=x'\oplus B.
        \end{equation}
        From the above information, we can compute the joint distribution of Alice's dits and Bob's guess conditioned on $b$. The procedure to obtain the joint distribution is very similar to the derivation for the uncorrelated case in \cite{our_prl}. The joint distribution for the van Dam protocol is given by,
		\begin{equation}\label{eq:joint_dist}
			\mathrm P(g=j,a_0=k,a_1=l|b=i)=\mathrm P(a_0=k,a_1=l)\sum_{m=0}^{d-1}\frac{1+(d-1)e^m_c}{d}\frac{1+(d-1)e^{j\oplus\overline{k}\oplus\overline{m}}_{\overline{k}\oplus l, i}}{d}.
		\end{equation}
		We now choose an input distribution such that Alice's dits are correlated and that their marginal distribution is uniform,
		\begin{equation}
			\sum_{k=0}^{d-1}q_{kl}=\sum_{l=0}^{d-1}q_{kl}=0.
		\end{equation}
		Now the statement of extended IC principle reads as follows,
		\begin{align}
			I(a_0;g \mid b=0)+I(a_1;g \mid b=1,a_0)\leq \mathcal{C}. 
		\end{align}
		Now we calculate the relevant marginal distributions and plug them into the IC equation. To make the structure clearer, we introduce the following notation
		
		\begin{align}
			f_{jkl}^i=\frac{(d-1)^2}{d}\sum_{m=0}^{d-1}e^m_ce^{j\oplus\overline{k}\oplus\overline{m}}_{\overline{k}\oplus l,i},\quad h^i_{jkl}=q_{kl}+(1+q_{kl})f^i_{jkl},\quad
			  h^i_{jk}=\frac1d\sum_{l=0}^{d-1}h^i_{jkl},\quad h^i_{j}=\frac{1}{d^2}\sum_{k,l=0}^{d-1}h^i_{jkl}.
		\end{align}
		Thus, we can write the joint distribution and the marginals as
		\begin{align}
			&\mathrm P(g=j,a_0=k,a_1=l|b=i)=\frac{1+h^i_{jkl}}{d^3}, \cr 
			&\mathrm P(g=j,a_0=k|b=i)=\frac{1+h^i_{jk}}{d^2},\quad \mathrm P(g=j|b=i)=\frac{1+h^i_{j}}{d}
		\end{align}
		Now the LHS of the IC statement looks as follows,
		
		\begin{align}
			\hspace{-1cm}I(a_0;g \mid b=0)+I(a_1;g \mid b=1,a_0) &= H(a_0,a_1)+H(g|b=0)+H(g,a_0|b=1)-H(g,a_0|b=0)-H(g,a_0,a_1|b=1)\cr
			&=-\frac{1}{d^2}\sum_{k,l=0}^{d-1}(1+q_{kl})\ln(1+q_{kl})-\frac{1}{d}\sum_{j=0}^{d-1}(1+h^0_{j})\ln(1+h^0_{j})\cr
			&-\frac{1}{d^2}\sum_{j,k=0}^{d-1}(1+h^1_{jk})\ln(1+h^1_{jk})+\frac{1}{d^2}\sum_{j,k=0}^{d-1}(1+h^0_{jk})\ln(1+h^0_{jk})\cr
			&+\frac{1}{d^3}\sum_{j,k,l=0}^{d-1}(1+h^1_{jkl})\ln(1+h^1_{jkl}).
		\end{align}
		The RHS is simply the capacity of the communication channel and is given by
		\begin{equation}
			\mathcal{C}=\ln d+\sum_{m=0}^{d-1}\frac{1+(d-1)e^m_c}{d}\ln \left(\frac{1+(d-1)e^m_c}{d}\right).
		\end{equation}
		Now we make the replacement $e^m_c\rightarrow \gamma e^m_c$. This essentially causes the tensors $h^i_{jkl},h^i_{jk}$ and $h^i_{j}$ to pick an additional factor of $\gamma$ in front of the corresponding $f^i_{jkl}$ terms. We now notice in the limit $\gamma\rightarrow 0$ the capacity vanishes. After driving the capacity to zero and using the L'Hopital's rule twice (the first derivative also approaches zero as we take the limit of $\gamma\rightarrow 0$),we get the following inequality (see \cite{our_prl} for the exact algorithm),
		\begin{equation}
			-\sum_{j=0}^{d-1}(h^0_{j})^2+\frac{1}{d}\sum_{j,k=0}^{d-1}(h^0_{jk})^2-(h^1_{jk})^2+\frac{1}{d^2}\sum_{j,k,l=0}^{d-1}(1+q_{kl})(f^1_{jkl})^2\leq (d-1)^2 \sum_{m=0}^{d-1}(e^m_c)^2.
		\end{equation}
This completes the proof.
        
\subsection{Details of the correlation slice}\label{app:slice_details} 

We provide the details of the inequalities and the correlation slice used for the comparison in \cref{fig:comparison}. We first reproduce the original IC inequalities from \cite{gachechiladze2022quantum,our_prl}.

\begin{theorem}[From \cite{gachechiladze2022quantum}]\label{res:d2dd}
In the $d2dd$ scenario, the following family of inequalities follows from the Information Causality principle:
\begin{align}
\sum_{i=0}^{1}
\left|
\sum_{j=0}^{d-1}\sum_{k=0}^{d-1}
\tilde e^{i\cdot j\oplus k}_{j,i}\omega^{k\cdot l}
\right|^2
\leq d^4,
\qquad
\forall l\in \{1,\dots,\lfloor d/2\rfloor\},
\end{align}
where
\[
\tilde e^{k}_{j,i}
=
d\,\mathrm{P}(A\oplus B=k \mid \alpha=j,\beta=i)-1,
\]
for $k,j\in [d]$ and $i\in\{0,1\}$, $\omega=e^{2\pi \mathrm{i}/d}$ is the $d$th root of unity, and $\oplus$ denotes addition modulo $d$.
\end{theorem}

The theorem above is stated using the bias convention of \cite{gachechiladze2022quantum}. In the computations below, we use the normalized convention of \cref{eq:input_dist}, namely
\[
e^k_{ij}
=
\frac{
d\,\mathrm{P}(A\oplus B=k\mid \alpha=i,\beta=j)-1
}{d-1}.
\]

We compare the above inequality with the inequality derived in \cref{thm:corr_ineq} on the same correlation slice used in the main text. The relevant extremal nonlocal and local distributions are
\begin{align}
\mathrm{P}^{\mathrm{PR}}_{ijk}(A,B\mid\alpha,\beta)
&=
\frac{1}{d}
\delta_{A\oplus B,\alpha\cdot\beta\oplus i\cdot\alpha\oplus j\cdot\beta\oplus k},\\
\mathrm{P}^{\mathrm{L}}_{ijkl}(A,B\mid\alpha,\beta)
&=
\delta_{A,i\cdot \alpha\oplus j}
\delta_{B,k\cdot \beta\oplus l}.
\end{align}
We choose the following specific mixture of one PR box, a symmetrized family of local boxes, and white noise:
\begin{align}\label{eq:app_slice_dist}
\mathrm{P}
=
(1-s-t)\mathrm{P}^{\mathrm{PR}}_{000}
+
\frac{s}{d}\sum_{r=0}^{d-1}\mathrm{P}^{\mathrm{L}}_{2r0\overline r}
+
t\,\mathrm{P}^{\mathrm{wn}},
\end{align}
where $\mathrm{P}^{\mathrm{wn}}(A,B\mid\alpha,\beta)=\frac{1}{d^2}$, $s,t\geq0$, $s+t\leq1$, and $\overline r\equiv -r \pmod d$ denotes the additive inverse of $r$.

The normalized biases are computed as
\begin{equation}
e^k_{ij}
=
\frac{
d\sum_{a=0}^{d-1}
\mathrm{P}(A=a,B=\overline a\oplus k\mid\alpha=i,\beta=j)-1
}{d-1}.
\end{equation}
We choose a specific channel by setting the biases according to
\begin{align}
e_c^m=
\begin{cases} 
-\frac{1}{d-1}, & m< \frac{d-1}{2},\\
\frac{1}{d-1}, & m> \frac{d-1}{2},\\
0, & \text{otherwise.}
\end{cases}
\end{align}
For Alice's input distribution, we choose
\begin{equation}
\mathrm{P}(a_0=k,a_1=l)=\frac{1+q_{kl}}{d^2},
\qquad
q_{kl}
=
\epsilon\,\operatorname{Re}\!\left[
\frac{1}{d}
+
\exp\!\left(
\frac{2\pi \mathrm{i}(k+1)(l+1)}{d+1}
\right)
\right].
\end{equation}

We fix $d=3$. After computing the biases and substituting the channel parameters and input distribution, we obtain the following inequalities:
\begin{align}\label{eq:app_slice_ineq}
I_{\mathrm{original}}:\quad
&- 1 + 3 s^2 + 4 s (t-1) + 2 (t-1)^2\leq 0,\\
I_{\mathrm{correlated}}:\quad
& - 1 + 3 s^2 + 4 s (t-1) + 2 (t-1)^2 \nonumber\\
&\quad
-\frac{1}{9} \epsilon \left[
2 (t-1)^2\epsilon
+2 s (t-1) (\epsilon+6)
+s^2 (\epsilon+12)
\right] \leq0.
\end{align}
To optimize over $\epsilon$, we compute the envelope of the correlated inequalities. Since the inequality is quadratic in $\epsilon$, the boundary of the envelope is obtained by eliminating $\epsilon$, equivalently by setting the discriminant of the quadratic polynomial in $\epsilon$ to zero. This gives
\begin{equation}
I_{\mathrm{optimized}}:\quad
7 s^4 + 18 s^3 (t-1)
+s^2 (20 t^2 - 40 t-19)
+2 s (6 t^3 - 18 t^2+ 17 t  -5 )
+2 (t-1)^2 ( 2 t^2-4t+1)\leq0.
\end{equation}

We plot both inequalities in this slice of the correlation space and observe that the correlated inequality is strictly stronger than the original IC inequality; see \cref{fig:comparison}.

\section{Proof for dit EARAC bounds from IC}\label{app:rac_bounds}
		We begin by specifying the scenario, Alice receives $n$ dits which are uniformly distributed and Bob receives $b\in[n]$. Since we are working in the case of the ``unbiased errors case", the probability of guessing the dit $a_i$ is independent of $i$ and thus, all the mutual information terms in the IC statement are identical \cite{miklin2021information}, which we denote by  
		\begin{equation}
			p=\frac{1}{n}\sum_{i=0}^{n-1}\text{P}(g=a_i|b=i)=\frac{1+(d-1)e_ce}{d}.
		\end{equation}
        
		Using Fano's inequality we have
		\begin{equation}
			n \left(\log d -h\left(\frac{1+(d-1)ee_c}{d}\right)-\frac{(d-1)(1-ee_c)}{d}\log(d-1)\right)\leq \sum_{i=0}^{n-1}I(g;a_i|b=i),
		\end{equation}
        where $h(x)=-x\log x-(1-x)\log(1-x)$ is the binary entropy function.
        
		We use a symmetric unbiased channel for which with probability $p_c=\frac{1+(d-1)e_c}{d}$ the message symbol is unchanged and with a probability $1-p_c=\frac{(d-1)(1-e_c)}{d}$ it changes into one of the other $d-1$ symbols. For such a channel we can express the capacity and divide both sides of the inequality by it, 
		\begin{align}
			n\frac{ \log d -h\left(\frac{1+(d-1)ee_c}{d}\right)-\frac{(d-1)(1-ee_c)}{d}\log(d-1)}{ \log d -h\left(\frac{1+(d-1)e_c}{d}\right)-\frac{(d-1)(1-e_c)}{d}\log(d-1)}\leq 1.
		\end{align}

		Taking the limit $e_c\rightarrow0$ can be computed by applying L’Hôpital’s rule twice, resulting in
		\begin{equation}
			e\leq\frac{1}{\sqrt n}.
		\end{equation}
		Therefore, the winning probability satisfies
\begin{equation}
    p\leq
    \frac{1}{d}\left(1+\frac{d-1}{\sqrt n}\right).
\end{equation}

\section{Details of optimality of EARAC bounds}\label{app:optimality_proof}
    
    \subsection{Proof of \cref{thm:corr_weak}}\label{app:optimality_proof1}

    We choose an isotropic family of PR-boxes with biases defined by
    \begin{equation}
        e^k_{ij}=[k= i\cdot j]e-[k\neq i\cdot j]\frac{e}{d-1}.
    \end{equation}

    We assume Alice has an input distribution of the form $\mathrm P(a_0=k,a_1=l)=\frac{1+\epsilon q_{kl}}{d^2}$. We work in the standard scenario where Alice and Bob use the generalized van Dam protocol. We then  compute the joint distribution and some relevant marginal distributions for the isotropic family of boxes from \cref{eq:joint_dist} below
    \begin{align}\label{eq:joint_dist_isotropic}
        \mathrm P(g=j|a_0=k,a_1=l,b=i)&=\sum_{m=0}^{d-1}\frac{1+(d-1)e^m_c}{d}\frac{1+(d-1)e^{j\oplus\overline{k}\oplus\overline{m}}_{\overline{k}\oplus l, i}}{d}\cr
        &=\frac{1+(d-1)ee_c^{\overline{k}\oplus j\oplus \overline{(\overline{k}\oplus l})\cdot i}}{d}=\begin{cases}
            \frac{1+(d-1)ee_c^{j\oplus \overline{k}}}{d} & i=0,\\
            \frac{1+(d-1)ee_c^{j\oplus \overline{l}}}{d} & i=1.
        \end{cases}\cr
       \mathrm P(g=j|a_0=k,b=i)&=\begin{cases}
            \frac{1+(d-1)ee_c^{j\oplus \overline{k}}}{d} & i=0,\\
            \frac{1+(d-1)e \epsilon c_{jk}}{d} & i=1,
        \end{cases}\cr
       \mathrm P(g=j|b=0)&= \frac{1+(d-1)e \epsilon b_j}{d},
    \end{align}
    where we define $q_k= \sum_{l=0}^{d-1}q_{kl}$, $\tilde q_l= \sum_{k=0}^{d-1}q_{kl}$, $b_j=\frac{1}{d^2}\sum_{k=0}^{d-1}q_{k}e_c^{j\oplus \overline{k}}$ and $c_{jk}=
\frac{\sum_{l=0}^{d-1}q_{kl}e_c^{j\oplus \overline{l}}}{d+\epsilon q_k}$ as some auxiliary variables. We define the total information as  $I_\epsilon(e,\{e_c^i\}_{i=0}^{d-1}):=I(g;a_0|b=0)+I(g;a_1|b=1,a_0)$. Expanding it in terms of the relative entropies between the variables we obtain 
    \begin{align}
    I_\epsilon(e,\{e_c^i\}_{i=0}^{d-1})&=I(g;a_0|b=0)+I(g;a_1|b=1,a_0)\cr
    &=H(g|b=0)-H(g|a_0,b=0)+H(g|a_0,b=1)-H(g|a_0,a_1,b=1)\cr
    &=H(g|b=0)-\sum_{k=0}^{d-1} \mathrm P(a_0=k)\left(H(g|a_0=k,b=0)-H(g|a_0=k,b=1)\right)\cr&-\sum_{k,l=0}^{d-1}P(a_0=k,a_1=l)H(g|a_0=k,a_1=l,b=1).
    \end{align}
    Plugging in the distributions from above we have
    \begin{align}
        I_\epsilon(e,\{e_c^i\}_{i=0}^{d-1})&=H\left(\frac{1+(d-1)e \epsilon b_j}{d}\right)-\sum_{k,l=0}^{d-1}\frac{1+\epsilon q_{kl}}{d^2}H\left(\frac{1+(d-1)ee_c^{j\oplus \overline{l}}}{d}\right)\cr
        &-\sum_{k=0}^{d-1}\frac{d+\epsilon q_k}{d^2}\left[H\left(\frac{1+(d-1)ee_c^{j\oplus \overline{k}}}{d}\right)-H\left(\frac{1+(d-1)e\epsilon c_{jk}}{d}\right)\right],
    \end{align}
    where $H(p^j)$ (with some abuse of notation) refers to $H(\{p^j\}_{j=0}^{d-1})=-\sum_j p^j \log_2 p^j$ i.e. the Shannon entropy of the distribution $\{p^j\}_{j=0}^{d-1}$ where $j$ is the index to be summed over. Putting $\epsilon=0$ we have,
    \begin{equation}
        I_0(e,\{e_c^i\}_{i=0}^{d-1})=2\left[\log_2 d -H\left(\frac{1+(d-1)ee_c^j}{d}\right)\right]
    \end{equation}

    If we now consider the difference between the correlated and uncorrelated total information, we have
\begin{align}
\hspace{-1cm}
&I_\epsilon(e,\{e_c^i\}_{i=0}^{d-1})
-
I_0(e,\{e_c^i\}_{i=0}^{d-1}) \nonumber\\
&=
H\left(\frac{1+(d-1)e\epsilon b_j}{d}\right)
-\frac{\epsilon}{d^2}\sum_{l=0}^{d-1}
\tilde q_l
H\left(\frac{1+(d-1)ee_c^{j\oplus \overline l}}{d}\right) \nonumber\\
&\quad
-\frac{\epsilon}{d^2}\sum_{k=0}^{d-1}
q_k
H\left(\frac{1+(d-1)ee_c^{j\oplus \overline k}}{d}\right)
+\sum_{k=0}^{d-1}
\frac{d+\epsilon q_k}{d^2}
H\left(\frac{1+(d-1)e\epsilon c_{jk}}{d}\right)
-2\log_2 d .
\label{eq:app_thm4_o1}
\end{align}

Since the entropy of a probability vector is invariant under cyclic shifts of its entries, the second and third terms vanish after summing over $l$ and $k$, respectively. Indeed, the corresponding entropy terms are independent of $l$ and $k$, while
\[
\sum_{l=0}^{d-1}\tilde q_l=0,
\qquad
\sum_{k=0}^{d-1}q_k=0.
\]
Therefore,
\begin{align}
&I_\epsilon(e,\{e_c^i\}_{i=0}^{d-1})
-
I_0(e,\{e_c^i\}_{i=0}^{d-1}) \nonumber\\
&=
\left(
H\left(\frac{1+(d-1)e\epsilon b_j}{d}\right)-\log_2 d
\right)
+
\left(
\sum_{k=0}^{d-1}
\frac{d+\epsilon q_k}{d^2}
H\left(\frac{1+(d-1)e\epsilon c_{jk}}{d}\right)
-\log_2 d
\right).
\end{align}
Using concavity of entropy, we can simplify the second summand as

\begin{align}
&\sum_{k=0}^{d-1}
\frac{d+\epsilon q_k}{d^2}
H\left(\frac{1+(d-1)e\epsilon c_{jk}}{d}\right) \leq
H\left(
\sum_{k=0}^{d-1}
\frac{d+\epsilon q_k}{d^2}
\frac{1+(d-1)e\epsilon c_{jk}}{d}
\right) =
H\left(
\frac{
1+(d-1)e\epsilon
\left(
\frac{1}{d^2}
\sum_{k,l=0}^{d-1}q_{kl}e_c^{j\oplus\overline l}
\right)
}{d}
\right).
\label{eq:app_thm4_o2}
\end{align}

Substituting this bound into the expression for
$I_\epsilon-I_0$, we obtain
\begin{align}
&I_\epsilon
-
I_0
\leq
\left(
H\left(\frac{1+(d-1)e\epsilon b_j}{d}\right)-\log_2 d
\right) \quad+
\left(
H\left(
\frac{
1+(d-1)e\epsilon
\left(
\frac{1}{d^2}
\sum_{k,l=0}^{d-1}q_{kl}e_c^{j\oplus\overline l}
\right)
}{d}
\right)
-\log_2 d
\right).
\end{align}
Since the Shannon entropy of any $d$-outcome probability distribution is at most $\log_2 d$, both terms in parentheses are nonpositive. Therefore,
\begin{equation}
I_\epsilon(e,\{e_c^i\}_{i=0}^{d-1})
\leq
I_0(e,\{e_c^i\}_{i=0}^{d-1}).
\end{equation}
Thus, introducing correlations among Alice's inputs cannot increase the total information for isotropic NS boxes, and the strongest bound is obtained for the uncorrelated input distribution.

\subsection{Proof of \cref{thm:opt_channel}}\label{app:optimality_proof2}
		We will now prove that the optimum bound on $e$ occurs for a uniform error channel. We begin by writing down the explicit form of the IC inequality for the uncorrelated case as a function of the box biases and the channel parameters i.e. we define $f(e,\{e_c^i\}_{i=0}^{d-1}):=I_0(e,\{e_c^i\}_{i=0}^{d-1})-\mathcal{C}$ as follows,
		\begin{equation}\label{eq:ic_mixing_uncorrelated}
			f(e,\{e_c^i\}_{i=0}^{d-1})=\frac{1}{d\log 2}\sum_{i=0}^{d-1} \left[ 2(1+(d-1)e e_{c}^i)\log(1+(d-1)e e_{c}^i)-(1+(d-1)e_{c}^i)\log(1+(d-1) e_{c}^i)\right].
		\end{equation}
        	We then define an auxiliary function $G$ such that
	\begin{equation}
		G(x)=2\left(1+(d-1)ex\right)\log\left(1+(d-1)ex\right)-\left(1+(d-1)x\right)\log\left(1+(d-1)x\right).
	\end{equation}
	Using the above definition, the function $f$ is then simply
	\begin{equation}
		f(e,\{e_c^i\}_{i=0}^{d-1})=\frac{1}{d\log 2}\sum_{i=0}^{d-1}G(e_c^i).
	\end{equation}
	Thus, to find the best mixing bound, we have the following optimization problem
    
    \begin{equation}\label{eq:optimisation}
        \begin{aligned}
        \max_{e,\{e_c^i\}_{i=0}^{d-1}}\quad & e\\
        \textrm{s.t.}\quad
            & f(e,\{e_c^i\}_{i=0}^{d-1})\leq0,\\
            & \sum_{i=0}^{d-1}e_c^i=0,\\
            & 0\leq e\leq \frac{1}{\sqrt{2}},\qquad -\frac{1}{d-1}\leq e_c^i\leq1\quad \forall i\in[d].
        \end{aligned}
    \end{equation}
     In the above, we took the value of $e\leq\frac{1}{\sqrt{2}}$, since this bound on the isotropic $d2dd$ family can be obtained from the quantum Bell inequalities derived in Refs.~\cite{gachechiladze2022quantum, our_prl}. These results use different normalization convention, $\tilde e^k_{j,i}=d\,\mathrm{P}(A\oplus B=k\mid\alpha=j,\beta=i)-1=(d-1)e^k_{j,i}$. Substituting the isotropic biases gives
\begin{equation}
\tilde e^{i\cdot j\oplus k}_{j,i}
=
\begin{cases}
(d-1)e, & k=0,\\
-e, & k\neq0,
\end{cases}
\end{equation}

    Thus, for each fixed $i\in\{0,1\}$, $\omega=e^{\frac{2\pi \mathrm{i}}{d}}$ and each nonzero Fourier mode $l$,
$
\sum_{j=0}^{d-1}\sum_{k=0}^{d-1}
\tilde e^{i\cdot j\oplus k}_{j,i}\omega^{kl}=
d^2e.
$
	 Thus, from the quantum Bell inequalities, we obtain $\sum_{i=0}^{1}|d^2e|^2\leq d^4$, which itself implies $2d^4e^2\leq d^4$, and hence $e\leq1/\sqrt2$. Thus the optimum is not attained at the boundary $e=1$.

    We can use the Karush–Kuhn–Tucker (KKT) approach to characterize the behavior of the optimum point. First, we write the Lagrangian,
	\begin{equation}
			\mathcal{L}(e,\{e_c^i\}_{i=0}^{d-1},\mu,\lambda,\{\alpha_i\}_{i=0}^{d-1},\{\beta_i\}_{i=0}^{d-1})=e+\mu f(e,\{e_c^i\})+\lambda\sum_{i=0}^{d-1}e_c^i + \sum_{i=0}^{d-1} \alpha_i \left( -\frac{1}{d-1}-e_c^i\right) + \sum_{i=0}^{d-1} \beta_i (e_c^i - 1)
	\end{equation}
	where $\mu \geq 0$ and $\alpha_i, \beta_i \geq 0$ are the dual variables for the inequality constraints, and $\lambda$ is the multiplier for the equality constraint. 
    Since $e=0$ is trivial and $e=1$ corresponds to a perfect PR box, which violates IC for every nontrivial channel, the relevant optimum satisfies $0<e<1$. We therefore treat $e$ as an interior variable in the KKT analysis.
    For an optimum, the KKT conditions must be satisfied as follows:
    \begin{align}
        &\frac{\partial\mathcal{L}}{\partial e}=1+\mu \sum_{i=0}^{d-1}\frac{\partial G}{\partial e}=0, \quad\frac{\partial\mathcal{L}}{\partial e_c^i}=\mu G'(e_c^i)+\lambda-\alpha_i+\beta_i=0,\qquad \forall i\in[d], \label{eq:app_KKT_stat}\\
        & f(e,\{e_c^i\}_{i=0}^{d-1})\leq0,\quad \sum_{i=0}^{d-1}e_c^i=0, \quad 0\leq e\leq \frac{1}{\sqrt{2}},\quad -\frac{1}{d-1}\leq e_c^i\leq1\quad \forall i\in[d],\label{eq:app_KKT_primal}\\
        &\mu\geq0, \quad\alpha_i,\beta_i\geq 0 \;\forall i\in[d],\label{eq:app_KKT_dual}\\
        & \mu f(e,\{e_c^i\}_{i=0}^{d-1})=0, \quad\alpha_i\left( \frac{1}{d-1}+e_c^i\right)=\beta_i (e_c^i - 1)=0 \; \forall i\in[d].\label{eq:app_KKT_slack}
    \end{align}
    Here, \cref{eq:app_KKT_stat,eq:app_KKT_primal,eq:app_KKT_dual,eq:app_KKT_slack} refer to the stationary, primal feasibility, dual feasibility and complementary slackness conditions respectively. For convenience, we absorbed the constant prefactor $1/(d\log 2)$ into the multiplier $\mu$. 
    We now consider the feasible points in two cases: those that are strictly interior and boundary ones. We consider the strictly interior points first. For them, we have that $-1/(d-1)<e_c^i<1$, and from \cref{eq:app_KKT_slack}, it follows $\alpha_i=\beta_i=0$ for all $i$. 
	Let us define another function $K(x)=G'(x)=(d-1)[2e\left(1+\log\left(1+(d-1)ex\right)\right)-1-\log\left(1+(d-1)x\right)]$. Using that, we can write the optimality condition for  as
	\begin{equation}\label{eq:optimality_bounds}
		K(e_c^i)=\frac{-\lambda}{\mu}\quad \forall i\in[d].
	\end{equation}
	For an interior optimum with respect to $e$, $\mu$ cannot vanish: if $\mu=0$, then $\partial\mathcal L/\partial e=1\neq0$, contradicting stationarity. Hence $\mu\neq0$ and hence from complementary slackness $f(e,\{e_c^i\})=0$.  Thus, \cref{eq:optimality_bounds} implies that the set of values $\{e_c^i\}_{i=0}^{d-1}$ is such that the function $K$ attains the same value for all of them. One obvious possible way is when all of the biases $e_c^i$ are identical. But since all the biases sum up to zero, it readily implies that all of them are zero, and that is a trivial solution. Thus, to have a non-trivial solution, at least two of the biases should be distinct. Now, considering the derivative of $K$, we notice that it has a unique solution
	\begin{equation}
		K'(x)=(d-1)^2\left(\frac{2e^2}{1+(d-1)ex}-\frac{1}{1+(d-1)x}\right)=0\implies x=\frac{2e^2-1}{(d-1)e(1-2e)}.
	\end{equation}
    Thus, the function $K$ can have only one extreme point, and hence such a function can be identically valued at most at two values of $x$. This implies that at most two of the biases can be distinct. 

    Let $m$ of the biases be equal to $r>0$ and the remaining $d-m$ biases be equal to $r'=-\frac{mr}{d-m}<0.$
    For fixed $m$, feasibility implies
    \begin{equation}
        0<r\leq \frac{d-m}{m(d-1)}.
    \end{equation}
    On the active IC boundary we have
    \begin{equation}\label{eq:app_isobound}
    u(e,m,r):=mG(r)+(d-m)G(r')=0.    
    \end{equation}
    If we differentiate $u$ w.r.t $e$ we have
    \begin{equation}
        \frac{\partial u}{\partial e}=2(d-1)mr\log\left(\frac{1+(d-1)er}{1+(d-1)er'}\right)\geq 0.
    \end{equation}
    For fixed $m$ and $r$, we see that $u$ is strictly increasing in $e$, and hence \cref{eq:app_isobound} defines a unique critical value $e^\ast(m,r)$. Implicit differentiation gives
\[
\frac{\partial e^\ast}{\partial m}
=
-\frac{\partial u/\partial m}{\partial u/\partial e}\bigg|_{e=e^\ast}.
\]
A direct calculation gives
\[
\frac{\partial u}{\partial m}
=
\frac{r-r'}{r'}J(r'),
\qquad
J(x):=G(x)-xG'(x).
\]
Since $r>0$ and $r'<0$, we have $(r-r')/r'<0$. Moreover,
\[
J'(x)=-xG''(x).
\]
If we compute $G''(x)$ we get,
\begin{equation}
    G''(x)=\frac{(d-1)^2(2e^2-1+(2e-1)(d-1)ex)}{(1+(d-1)x)(1+(d-1)ex)}.
\end{equation}
We notice that the numerator in the above expression is linear in $x$ and achieves a maximum value of $2e^2-1$ when $e\geq1/2$ or $e-1$ when $e\leq 1/2$, respectively. Hence for $x\in[-1/(d-1),0]$ and $e\leq1/\sqrt2$, one has $G''(x)\leq0$, and hence $J'(x)\leq0$. Thus, $J$ is nonincreasing on this interval. Since $J(0)=0$ and $r'<0$, it follows that $J(r')\geq0$. Therefore $\partial u/\partial m\leq0$, and since $\partial u/\partial e>0$, we obtain
\[
\frac{\partial e^\ast}{\partial m}\geq0.
\]
Thus $e^\ast(m,r)$ is pointwise nondecreasing in $m$. Since the admissible domain $\left[0,\frac{d-m}{m(d-1)}\right]$
shrinks as $m$ increases, minimizing over $r$ cannot reverse this monotonicity. Hence, the tightest bound is attained for the smallest multiplicity, $m=1$.

    We now consider boundary points. There are two possible types of active boundary constraints:
    \begin{equation}
    e_c^i=1
    \qquad\text{or}\qquad
    e_c^i=-\frac{1}{d-1}.
    \end{equation}

    First, suppose that $e_c^{i_0}=1$ for some index $i_0$. Then the normalization condition gives $ \sum_{i\neq i_0}e_c^i=-1. $
    Since each remaining bias satisfies $e_c^i\geq -1/(d-1)$, the only way this equality can hold is if $ e_c^i=-\frac{1}{d-1},\  \forall i\neq i_0.$ Thus, this boundary case is already of the uniform-error form,
    \begin{equation}
    \{e_c^i\}_{i=0}^{d-1}
    =
    \left\{
    1,
    \underbrace{-\frac{1}{d-1},\dots,-\frac{1}{d-1}}_{d-1\ \mathrm{times}}
    \right\},
    \end{equation}
    up to permutations.

    It remains to consider boundary points for which some biases are equal to $-1/(d-1)$, but no bias is equal to $1$. We claim that such a point cannot be optimal unless all negative biases are equal. Let $s-m $ of the biases be equal to $\frac{-1}{d-1}$. From the KKT conditions for the remaining biases, they can assume at most two distinct values. Thus, the biases look like
    \begin{equation}
    \{e_c^i\}_{i=0}^{d-1}
    =
    \left\{
    \underbrace{-\frac{1}{d-1},\dots,-\frac{1}{d-1}}_{s-m\ \mathrm{times}},\underbrace{r',\dots,r'}_{m\ \mathrm{times}},\underbrace{r,\dots,r}_{d-s\ \mathrm{times}}
    \right\},
    \end{equation}
    where $r>0,r'<0$ as before. Let $x_1,\dots,x_s<0$ be the negative channel biases. We remind ourselves that for $e\leq1/\sqrt2$, we have $ G''(x)\leq0$ for all $x\in\left[-\frac{1}{d-1},0\right]$,
    so $G$ is concave on the negative part of the feasible interval. Replacing each of the $s$ biases with their average
    \begin{equation}
    \bar x=\frac{1}{s}\sum_{\ell=1}^s x_\ell,
    \end{equation}
    preserves the normalization condition and remains feasible, since
    $-\frac{1}{d-1}\leq \bar x<0$. By concavity of G,
    \begin{equation}
    \sum_{\ell=1}^s G(x_\ell)
    \leq
    sG(\bar x).
    \end{equation}
    This implies that 
    \begin{equation}
        f(e,\{x_1,\dots,x_s,r,\dots,r\})\leq f(e,\{\bar x,\dots,\bar x,r,\dots,r\}).
    \end{equation}
     Therefore, this averaging operation does not decrease $f=\sum_{i=0}^{d-1}G(e_c^i)$, and any bound on $e$ must be lower than the bound obtained from the original assignment since $f$ is increasing. Thus, at an optimal boundary point, all negative channel biases may be taken to be equal. In doing so, if a subset of them assume a boundary value of $\frac{-1}{(d-1)}$, the averaging procedure will always result in an interior point (unless all the negative biases are exactly $\frac{-1}{(d-1)}$, which reduces to the uniform-error channel case). 

    Consequently, any optimal boundary candidate reduces to a two-valued channel-bias pattern,
    \begin{equation}
    \left\{
    \underbrace{r,\dots,r}_{m\ \mathrm{times}},
    \underbrace{r',\dots,r'}_{d-m\ \mathrm{times}}
    \right\},
    \qquad
    r>0,\quad r'<0.
    \end{equation}
    This is precisely the two-valued case analyzed above. Since the critical value $e^\ast(m,r)$ is nondecreasing in $m$, the tightest bound is obtained for the smallest possible multiplicity, namely $m=1$.
    Thus, we obtain the uniform-error channel used in \cite{miklin2021information}, and the claim follows.

\newpage
\twocolumngrid 
\bibliography{bibliography.bib}

\end{document}